\documentclass[
final
]{dmtcs-episciences}

\usepackage[utf8]{inputenc}

\usepackage[utf8]{inputenc}

\usepackage{xstring}
\usepackage{textcomp}
\usepackage[mathscr]{eucal}
\usepackage{color}
\usepackage{xcolor}
\usepackage{enumerate}
\usepackage{graphicx}
\usepackage{array}
\usepackage{multirow}
\usepackage{tabularx}
\usepackage[online]{threeparttable}
\usepackage{listings}
\usepackage{subcaption}

\usepackage{xspace}
\usepackage{complexity}
\usepackage{setspace}
\usepackage{algorithm}
\usepackage{amsmath,amssymb,amsthm}
\usepackage{comment}
\usepackage{cleveref}
\usepackage{thm-restate}
\usepackage[noend]{algpseudocode}
\usepackage{csquotes} 
\usepackage{nicefrac}

\crefname{enumi}{condition}{conditions}

\newcommand{\parmacro}[1]{\vspace{0.5em}\textbf{#1}}

\newcommand{\opt}{\mathsf{opt}}

\newcommand{\pair}[1]{\langle #1 \rangle}

\DeclareMathOperator\scope{scope}
\newcommand{\ME}{\textsc{ME2C}\xspace}
\newcommand{\UG}{\textsf{UG}\xspace}

\newtheorem{lemma}{Lemma}
\newtheorem{definition}{Definition}

\author[T. M\"omke et al.]{Tobias~M\"omke\affiliationmark{1}\thanks{Partially supported by DFG Grant 439522729 (Heisenberg-Grant) and DFG Grant 439637648 (Sachbeihilfe).}
  \and Alexandru~Popa\affiliationmark{2}\thanks{Partially supported by a grant of the Ministry of Research, Innovation and Digitization, CNCS - UEFISCDI, project number PN-III-P1-1.1-TE-2021-0253, within PNCDI III.}
  \and Aida~Roshany-Tabrizi\affiliationmark{1}
  \\ \and  Michael~Ruderer\affiliationmark{1}
  \and Roland~Vincze\affiliationmark{1}}

\title[Approximating Maximum Edge 2-Coloring by Normalizing Graphs]{Approximating Maximum Edge 2-Coloring\\ by Normalizing Graphs\thanks{A preliminary version of this paper was presented at WAOA 2023.}}
\affiliation{
  University of Augsburg, Augsburg, Germany\\
  University of Bucharest, Bucharest, Romania}

\keywords{Approximation Algorithms \and Edge 2-Coloring \and Matchings}

\begin{document}

\publicationdata{vol. 27:2}{2025}{7}{10.46298/dmtcs.13212}{2024-03-12; 2024-03-12; 2024-10-03}{2025-03-03}

\maketitle            
\begin{abstract}
In a simple, undirected graph $G$, an edge $2$-coloring is a coloring of the edges such that no vertex is incident to edges with more than 2 distinct colors.
The problem maximum edge $2$-coloring (\ME) is to find an edge $2$-coloring in a graph $G$ with the goal to \emph{maximize} the number of colors.
For a relevant graph class, \ME models anti-Ramsey numbers and it was considered in network applications.
For the problem a $2$-approximation algorithm is known, and if the input graph has a perfect matching, the same algorithm has been shown to have a performance guarantee of $5/3 \approx 1.667$.
It is known that \ME is \APX-hard and that it is \UG-hard to obtain an approximation ratio better than $1.5$.
We show that if the input graph has a perfect matching, there is a polynomial time $1.625$-approximation and if the graph is claw-free or if the maximum degree of the input graph is at most three (i.e., the graph is subcubic), there is a polynomial time $1.5$-approximation algorithm for \ME.
\end{abstract}
\section{Introduction} 
In a simple, undirected graph $G$, an edge $2$-coloring is a coloring of the edges such that no vertex is incident to edges with more than 2 distinct colors.
The problem maximum edge $2$-coloring (\ME) is to find an edge $2$-coloring in $G$ that uses a  \emph{maximal} number of colors.
Formally, we aim to compute a coloring $\chi\colon E(G) \rightarrow \mathbb{N}$ that maximizes $|\{c \in \mathbb{N} \mid \chi(e) = c \text{ for an $e \in E(G)$}\}|$, such that for each vertex $v \in V(G)$, $|\{ c \in \mathbb{N} \mid \chi(e) = c \text{ for an $e$ incident to $v$}\}| \le 2$ holds. 

Maximum edge $2$-coloring is a particular case of anti-Ramsey numbers and has been considered in combinatorics.
For given graphs $G$ and $H$, the \emph{anti-Ramsey number} $\textrm{ar}(G,H)$ is defined to be the maximum number of colors in an edge-coloring that does not produce a rainbow copy of $H$ in $G$, i.e., a copy of $H$ in $G$ with every edge of $H$ having a unique color.
Classically, the graph $G$ is a large complete graph and $H$ is from a particular graph class.
If $H$ is a star with three leaves and $G$ is an arbitrary graph,
the anti-Ramsey number 
is precisely the maximum number of colors in an edge $2$-coloring.

The study of anti-Ramsey numbers was initiated by Erd\H{o}s, Simonovits and S\'os in 1975~\cite{ErdosSS75}.
Since then, there have been a large number of results on the topic including the cases where
$G = K_n$ and $H$ is a cycle~\cite{ErdosSS75,Montellano-Ballesteros2005,AxenovichJK04}, tree~\cite{JiangW04,Jiang2002}, clique~\cite{FriezeR93,ErdosSS75,BlokhuisFGR01}, matching~\cite{Schiermeyer04,ChenLT09,HaasY12} or a member of some other class of graphs~\cite{ErdosSS75,AxenovichJ04}. 

The main application of \ME comes from wireless mesh networks. 
Raniwala et al.~\cite{RaniwalaGC04,RaniwalaC05} proposed a wireless architecture in which each computer uses two network interface cards (NICs) compared to classical architectures that use only one NIC.
In this model, each computer can communicate with the other computers in the network using two channels.
Raniwala et al.~\cite{RaniwalaGC04,RaniwalaC05} showed that using such an architecture can increase the throughput by a factor of~$6$.
In order to minimize the interference, it is desirable to maximize the number of distinct channels used in the network.
In \ME computers correspond to nodes in the graph, while colors correspond to channels.
\subsection{Previous Work}

The problem of finding a maximum edge $2$-coloring of a given graph has been first studied by Feng et al.~\cite{FengZQW07,FengCZ08,FengZW09}.
They provided a $2$-approximation algorithm for ME2C and show that ME2C is solvable in polynomial time for trees and complete graphs, but they left the complexity for general graphs as an open problem.
The authors also studied a generalization of ME2C, the maximum edge $q$-coloring, where each vertex is allowed to be incident to at most $q$ edges with distinct colors.
For the maximum edge $q$-coloring they showed a $(1+\frac{4q-2}{3q^2-5q+2})$-approximation for $q > 2$.
 
Later, Adamaszek and Popa~\cite{AdamaszekPopa2016} showed that the problem is \APX-hard and proved that the algorithm above provides a $5/3$-approximation for graphs which have a perfect matching.
The \APX-hardness is achieved via a reduction from the Maximum Independent Set problem and states that maximum edge $2$-coloring problem is \UG-hard to approximate within a factor better than $1.5 - \epsilon$, for some $\epsilon > 0$.
Chandran et al.~\cite{CHJMRS23} showed that the matching-based algorithm of \cite{FengZW09} yields a $(1 +\frac{2}{\delta})$-approximation for graphs with minimum degree $\delta$ and a perfect matching.
If additionally the graph is triangle-free, the ratio improves to $(1 + \frac{1}{\delta - 1})$.
Recently, Chandran et al.~\cite{chandran2021improved} improved the analysis of the achieved approximation ratio for triangle-free graphs with perfect matching to $8/5$.
They also showed that the algorithm cannot achieve a factor better than $58/37$ on triangle free graphs that have a perfect matching. Dvořák and Lahiri~\cite{dvovrak2023maximum} designed a PTAS for the maximum edge q-colouring problem on minor free graphs.

Larjomaa and Popa~\cite{LarjomaaP15} introduced and studied the min-max edge $2$-coloring problem a variant of the ME2C problem, where the goal is to find an edge $2$-coloring that minimizes the largest color class.
Mincu and Popa~\cite{mincu2018heuristic} introduced several heuristic algorithms for the min-max edge $2$-coloring problem.

\subsection{Our results}

Our core algorithm, \Cref{alg:basic}, is the $2$-approximation algorithm for general graphs of Feng et al.~\cite{FengZW09}.
The algorithm simply finds a maximum matching, colors each edge of the matching with a distinct color, removes the edges of the matching and finally, colors each connected component of the remaining graph with a distinct color. 

Directly applying the algorithm, however, cannot provide an approximation ratio better than $2$ in the general case~\cite{FengZW09} and not better than $5/3 \approx 1.667$ for graphs with perfect matchings~\cite{AdamaszekPopa2016}.
To overcome this difficulty, we introduce a preprocessing phase which considerably simplifies the instance.
The simplifications both improve the quality of the solution provided by the algorithm \textit{and} lead to an improved upper bound on the size of an optimal solution.
A graph is called \textit{normalized} if no more preprocessing steps can be performed on it.

We first show that \Cref{alg:basic} is a $1.5$-approximation algorithm if \emph{after} the normalization, the graph contains a perfect matching (\Cref{thm:normalized_perfect_matching}).
We can ensure this property if we normalize a subcubic graph\footnote{Recall that a graph is subcubic if no vertex has a degree larger than three.}, even if before applying the normalization it did not have a perfect matching.

\begin{restatable}{theorem}{subcubic}\label{thm:subcubic}
\ME in subcubic graphs has a {polynomial-time} $1.5$-approximation algorithm.
\end{restatable}

It has been shown that claw-free graphs contain a perfect matching \cite{Sumner1974,LasVergnas1975}.
Some preprocessing steps might introduce claws which worsen the quality of the solution provided by \Cref{alg:basic}, therefore we do not immediately obtain a $1.5$-approximation.
However, we develop a bookkeeping technique to counteract this effect.

\begin{restatable}{theorem}{clawfree}\label{thm:clawfree}
  There is a {polynomial-time} $1.5$-approximation for claw-free graphs.
\end{restatable}

In the more general case of graphs with perfect matchings, the effect of introduced unmatched vertices is more severe.
We use a sophisticated accounting technique to quantify the effects of the appearing unmatched vertices on the quality of both the optimal solution and the solution given by \Cref{alg:basic}.
As a result we obtain a weaker but improved approximation algorithm for graphs containing a perfect matching: 

\begin{restatable}{theorem}{perfectmatching}\label{thm:perfect_matching}
  There is a $1.625$-approximation for graphs that contain a perfect matching.
\end{restatable}
Let us now elaborate on the key ideas behind our results.
After the preprocessing phase we obtain a normalized graph via a series of \textit{modifications}.
Intuitively, the modifications achieve the following: 1)  \emph{Avoid leaves with equal neighborhoods;} 2) \emph{Avoid degree-2 vertices;} 3) \emph{Avoid a specific class of triangular cacti.}

A triangular cactus is a connected graph such that two cycles have at most one vertex in common and each edge is contained in a $3$-cycle.
For our purposes, we additionally require that no vertex of the cactus is incident to more than one edge not in the cactus.

While the three modifications are relatively simple, proving that they are approximation-preserving is non-trivial.
If none of these modifications can be applied (anymore), we call the graph normalized.
Our key insight is that the number of colors in an optimal solution of a normalized graph can be bounded from above, this is stated as \Cref{lem:maxcolors} and shown in \Cref{sec:upper-bound}.

\begin{restatable}{lemma}{maxcolors}\label{lem:maxcolors}
  Let $G$ be a normalized connected graph with $n \ge 3$ vertices and $\ell$ leaves.
  Then there is no feasible coloring $\chi$ with more than $3n/4 - \ell/4$ colors.
\end{restatable}

We note that without normalization, an optimal solution can have $n$ colors (e.g., if $G$ is an $n$-cycle).
In order to prove \Cref{lem:maxcolors}, we use the notion of \emph{character graphs} introduced by Feng et al.~\cite{FengZW09}.
A character graph of an edge 2-coloring is a graph that contains exactly one edge from each color class. 
We first show that for a normalized graph we can ensure the existence of a \emph{nice} character graph, which is a character graph with several useful properties.
These properties allow for a counting argument with respect to the number of components in the character graph, which allows us to prove the bound in \Cref{lem:maxcolors}.

For general graphs, the best result is still the known $2$-approximation. 
There is a family of bipartite triangle free $2$-connected graphs with minimum degree $3$ which certifies this lower bound for our algorithm.
  
The rest of the paper is organized as follows. In \Cref{sec:algorithm} we describe the three modifications performed on the input graph before applying the algorithm.
Then, in \Cref{sec:upper-bound}, we prove the upper bound on the optimal solution on normalized instances.
In \Cref{sec:results} we combine the results from \Cref{sec:algorithm,sec:upper-bound} to prove \Cref{thm:subcubic,thm:clawfree}, and finally, in \Cref{sec:pm} we prove \Cref{thm:perfect_matching}.

\section{The Algorithm}
\label{sec:algorithm}

Let $G$ be a graph and $\chi$ a feasible $2$-coloring of the edges.
Recall that $\chi(e)$ marks the color of the edge $e$ in $\chi$.
With a slight abuse of notation, let us denote the set of all colors of the edges of~$G$ by $\chi(G)$ and the colors incident to a vertex $v$ by $\chi(v) := \{c \in \mathbb{N} \mid \exists u \in V(G): \chi(uv) = c\}$.
If a vertex $v$ is incident to an edge colored $c$, we say that vertex $v$ \emph{sees} $c$.
We also denote the number of colors in a coloring $\chi$ by $|\chi|$. 

For a color $c$, $E(c)$ denotes the set of edges with color $c$, that is, $E(c) := \{e \in E(G) \mid \chi(e) = c\}$.
We refer to $E(c)$ as the \emph{color class} of $c$.
Furthermore, we define by $V(c) :=  \{v \in V(G) \mid c \in \chi(v)\}$ the \emph{color class of $c$}, i.e. the vertices that see the color $c$.
Finally, $G(c) := (V(c), E(c))$ is the subgraph of $G$ whose edges have color $c$.
We call a cycle on 3 vertices a \emph{3-cycle} or \emph{triangle}.
The term \emph{pendant vertex} or \emph{leaf} is used for degree-$1$ vertices, while the term \emph{pendant edge} marks the edge incident to a pendant vertex.

\begin{algorithm}
\caption{The basic algorithm.\label{alg:basic}}
  \begin{algorithmic}[1]
    \Statex \textbf{Input:} A simple undirected graph $G=(V,E)$.
    \Statex \textbf{Output:} An edge 2-coloring $\chi$ on the edges of $G$.
    \State  Calculate a maximum cardinality matching $M$ in $G$.
    \State  Assign a distinct color in $\chi$ for every edge of $M$.
    \State  Assign a distinct color in $\chi$ for every nontrivial connected component of $E \setminus E(M)$.
  \end{algorithmic}
\end{algorithm}

While \Cref{alg:basic} is well studied (cf.~\cite{FengZW09,AdamaszekPopa2016,chandran2021improved}),
we apply some preprocessing steps to each problem instance $G$, before applying \Cref{alg:basic} to the resulting graph $G'$.
This preprocessing gives the graph $G'$ more structure, which will help us to prove better approximation guarantees.

These preprocessing steps consist of different \emph{modifications}, which will be defined throughout the paper.
We note that modifications can increase the size of the maximum matching~$M$, and therefore are not only for the analysis, but they change the instance in order to obtain stronger results.

Intuitively, a valid modification is a modification such that the number of colors in an optimal solution does not change and we can transform a solution for the modified instance to a solution for the original instance.
Formally, we define the following equivalence relation.
For a graph $G$ let $\opt(G)$ denote the number of colors in an optimal edge $2$-coloring  of $G$.
\begin{definition}
\label{def:equivalence}
  Two graphs $G = (V,E)$ and $G' = (V',E')$ form an \emph{equivalent pair} with respect to edge 2-coloring, denoted by $\pair{G, G'}$, if
  \begin{enumerate}
      \item \label{eqpair:1} an optimal edge 2-coloring of $G'$ uses the same number of colors as an optimal edge 2-coloring of $G$, i.e., $\opt(G) = \opt(G')$.
      \item \label{eqpair:2}For every edge 2-coloring $\chi'$ of $G'$ one can in polynomial time compute an edge 2-coloring~$\chi$ for $G$ that uses the same number of colors as $\chi'$, i.e., $|\chi| = |\chi'|$.
  \end{enumerate}
\end{definition}
To show $\pair{G,G'}$, it is sufficient to show $\opt(G) \le \opt(G')$ for \Cref{eqpair:1} and $|\chi| \ge |\chi'|$ for \Cref{eqpair:2}:
For \Cref{eqpair:1}, we use that there is a coloring $\chi'$ for $G'$ and a coloring $\chi$ for $G$ such that $\opt(G) \le \opt(G') = |\chi'| \le |\chi| \le \opt(G)$ and thus all inequalities have to be satisfied with equality.
For \Cref{eqpair:2}, we note that it is always possible to reduce the number of colors.

\begin{definition} \label{def:modification}
  A \emph{valid modification} is a sequence of vertex/edge alterations (additions or deletions), that result in a graph $G'$ such that $\pair{G, G'}$ is an equivalent pair.
\end{definition}

All modifications that will be introduced in the following are indeed \emph{valid} modifications. 

\begin{lemma}\label{lem:mod_approx}
Let $G$ be a graph and let $G'$ be the graph obtained from $G$ via a valid modification from Definition~\ref{def:modification}.
Given a polynomial time $\alpha$-approximation algorithm for $G'$ we can obtain a polynomial time $\alpha$-approximation algorithm for $G$.
\end{lemma}
\begin{proof}
An $\alpha$-approximation algorithm for $G'$ produces a coloring $\chi'$ with at least $\opt(G') / \alpha$ colors. Following \Cref{def:equivalence,def:modification} we know that $\opt(G) \leq \opt(G')$ and
that we can find in polynomial time a coloring $\chi$ for the graph $G$ that uses at least as many colors as $\chi'$.
Thus, we can find in polynomial time a coloring $\chi$ for $G$ that has at least $\opt(G) / \alpha$ colors, that is, we obtain an $\alpha$-approximation algorithm.
 \end{proof}

\paragraph*{Modification~1: Avoid pendant vertices with equal neighborhoods.}

The first modification is to remove a leaf $w$ from $G$, if there is another leaf $v$, such that both are incident to the same vertex $u$ (see \Cref{fig:leaves}).
Formally, we require that the following two conditions are simultaneously satisfied:
(i) the degree of $v$ in $G$ is one; and
(ii) there is a vertex $w \neq v$ of degree one and a vertex $u$ of degree at least three adjacent to both $v$ and~$w$.

We note that we require the degree constraint on $u$ to obtain a cleaner proof. 
If the degree equals $2$, we will see that
the following Modification~2 applies.

\paragraph*{Modification~2: Avoid degree-2 vertices.}

Given a vertex $v$ of degree $2$, we break it into two vertices $v_1$ and $v_2$ with degree $1$ each.
More precisely, let $u_1$ and $u_2$ be the two vertices adjacent to $v$. 
We replace the edges $u_1v$ and $u_2v$ by the edges $u_1v_1$ and $u_2v_2$, replacing $v$ by two new vertices $v_1$ and $v_2$.

\paragraph*{Modification~3: Remove Triangular Cacti.}

Recall that a triangular cactus is a connected graph such that two cycles have at most one vertex in common and each edge is contained in a $3$-cycle.\footnote{Note that our definition of a cactus is stricter than usual: we do not allow cut-edges.}
In other words it is a \enquote*{tree} of triangles where triangle pairs are joined by a single common vertex.
Let $C$ be a subgraph of $G$. 
For the modification we require that $C$ is a triangular cactus
such that for each vertex $v \in V(C)$, the degree of $v$ in $G$ is $3$ or $4$, and $v$ is incident to at most one edge from $E(G) \setminus E(C)$.
In the following, we call such a cactus a \emph{simple cactus}, and we call an edge $e \in E(G) \setminus E(C)$ a \textit{needle} of the cactus $C$.
Note that we do not require $C$ to be an \textit{induced} subgraph of $G$.
Indeed, two triangles can share a needle, as illustrated in \Cref{fig:simple_cactus}.

Modification~3 replaces each $3$-cycle by a single edge as follows:\\
Let $T_1,T_2,\dotsc,T_\ell$ be the triangles that comprise the simple cactus $C$.
For each $T_i$ with vertices $\{u_i,v_i,w_i\}$, it removes the edges $u_i v_i, v_i w_i, w_i u_i$ and introduces two new vertices $x_i$ and $y_i$ with an edge  $x_i y_i$. Finally it discards all isolated vertices (for an illustration, see the the full version of the article).

We now show that the graph obtained this way is equivalent to the original graph.

\begin{figure}[h!tb]
    \begin{center}
        \begin{subfigure}[h]{0.25\linewidth}
            \includegraphics[scale=0.4]{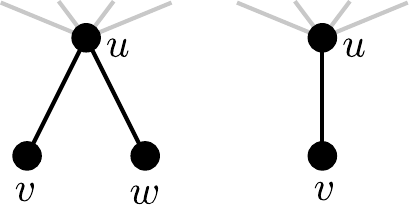}
        \end{subfigure}
        \begin{subfigure}[h]{0.4\linewidth}
            \includegraphics[scale=0.4]{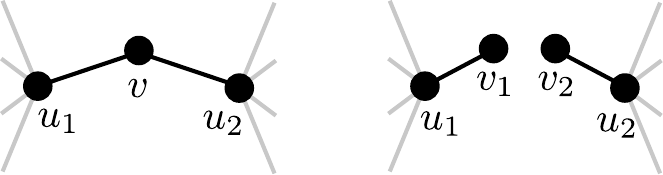}
        \end{subfigure}%
        \begin{subfigure}[h]{0.25\linewidth}
            \includegraphics[scale=0.3]{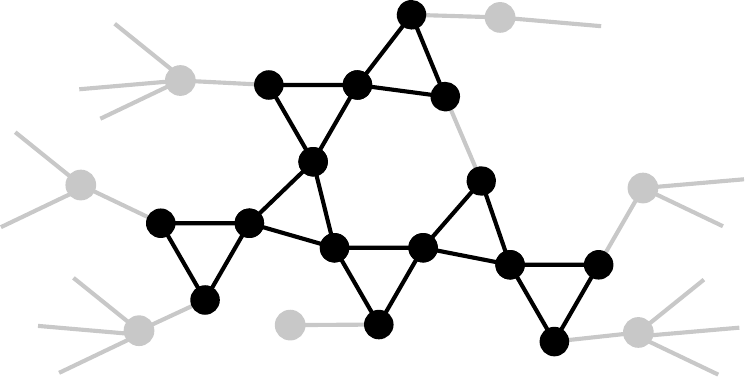}
        \end{subfigure}%
    \end{center}
\caption{Modifications 1 and 2,\label{fig:leaves} and a simple cactus (Modification 3) as a subgraph of $G$.\label{fig:simple_cactus}}
\end{figure}

\begin{restatable}{lemma}{mod-one-to-three}\label{lem:mod1}\label{lem:mod2}\label{lem:mod3}
    Modifications 1, 2 and 3 are valid.
\end{restatable}
\begin{proof}

\emph{Modification 1.} We observe that adding a leaf to a vertex $u$ cannot decrease the number of colors: we can simply color the new pendant edge with an arbitrary color seen by $u$.
    We therefore only have to show that removing a leaf does not decrease the number of colors.
    Let $G$ be a graph with an optimal coloring $\chi^\star$. 
    Suppose there exist vertices $v \neq w$ of degree one that are connected to the vertex $u$ of degree at least three.
    Let $x \notin \{u,v,w\}$ be a third vertex adjacent to $u$.
    By the definition of an edge 2-coloring, $u$ sees at most two colors.
    
    We distinguish between two cases.
\begin{itemize}

 \item   \textbf{Case 1}: $\chi^\star(uv) \neq \chi^\star(uw)$. Then $\chi(ux) \in \{\chi(uv),\chi(uw)\}$. Without loss of generality, we assume $\chi(ux) = \chi(uv)$.
    Then removing $uv$ from $G$ results in a graph $G'$ with coloring $\chi'$ such that $|\chi^\star| = |\chi'|$.
    
  \item   \textbf{Case 2}: $\chi^\star(uv) = \chi^\star(uw)$.
    Then again, removing $uv$ from $G$ does not change the number of colors
    and we obtain
    a feasible coloring $\chi'$ for $G'$ with $|\chi^\star| = |\chi'|$.    
\end{itemize}    

\paragraph*{Modification 2.}

    Let $G = (V, E)$ be a graph with a degree-$2$ vertex $v$ and let $G'$ be the graph obtained from $G$ by applying Modification~2 to $v$.
    Thus if $v$ is adjacent to $u_1, u_2 \in V$, then $G'$ contains the leaves $v_1$ and $v_2$ in place of $v$ as well as the new edges $v_1u_1$ and $v_2u_2$.

    Let $\chi^\star$ be an optimal coloring of $G$.
    Coloring $u_1 v_1$ with $\chi^\star(u_1 v)$ and $u_2 v_2$ with $\chi^\star(u_2 v)$ yields a feasible coloring $\chi'$ for $G'$ with the same number of colors, hence $\opt(G) \leq \opt(G')$.
    In the other direction, let $\hat{\chi}$ be an arbitrary coloring of $G'$. After contracting $v_1, v_2$ to a single vertex $v$,  $\hat{\chi}$ is still a feasible coloring of $G'$
    as $v$ only sees the colors of $\hat{\chi}(vu_1)$ and $\hat{\chi}(vu_2)$. 
Thus $|\chi'| = |\chi|$ and therefore Modification~2 is valid.

\paragraph*{Modification 3.}

\begin{figure}[h]
    \centering
    \includegraphics[scale=0.4]{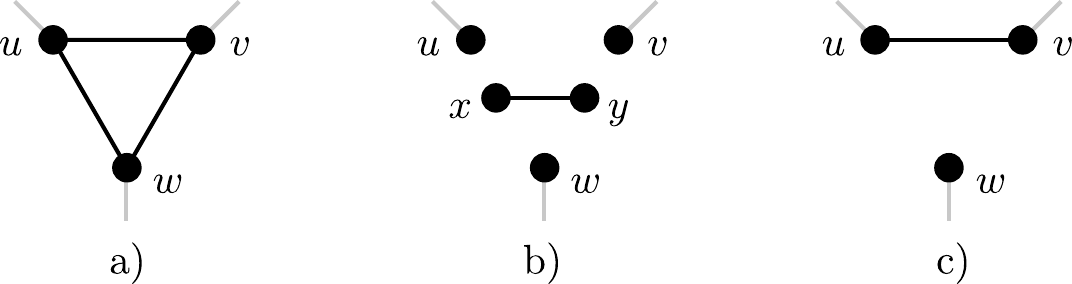}
    \caption{Subfigures a) and b): Modification~3 for a simple triangular cactus consisting of just one triangle. Subfigures a) and c) show a variant of Modification~3 that keeps the perfect matching (used in \Cref{sec:pm}).\label{fig:triangle}}
\end{figure}

Let $m$ be the number of adjacent triangles that form the cactus $C$ and let  $T_1, T_2,  \dots,  T_m$ be these triangles. Let $N$ be the set of needles.

After the modification, there are $m$ isolated edges and the needles can be colored independently of $C$.
In particular, all needles can still be colored with the same colors as before.
To show that we obtain an equivalent pair, we have to show that the modification does not decrease the optimal number of colors and that a coloring in the modified graph gives a coloring in the original graph without losing colors.
The latter condition is simple: We color the $m$ triangles monochromatically with the $m$ colors of the separate edges.
Since each vertex of $C$ either sees exactly two triangles or one triangle and one needle, the obtained coloring is feasible independent of the coloring of the needles.
To show the first condition, we define a \textit{novel} color to be a color \emph{not} appearing in $E(G) \setminus E(C)$.
Note that in particular, a needle cannot have a novel color. 
We then show that $C$ cannot have more than $m$ novel colors. Since after the modification we obtain exactly $m$ novel colors for the $m$ separate edges, we obtain an equivalent pair and therefore the validity of the modification follows.

We show by contradiction that $C$ cannot have more than $m$ novel colors.
Otherwise let $m$ be the smallest number of triangles such that there is a problem instance with a cactus $C$ of $m$ triangles and at least $m+1$ novel colors, and 
let $\chi$ be a feasible coloring in which $C$ has more than $m$ novel colors.
Let $G'$ be the subgraph of $G$ induced by the edges $E(C) \cup N$.

Then in particular, $\chi$ restricted to $G'$ is a feasible coloring.
We start by simplifying the instance.
Let $c$ be a color that is not novel.
We observe that within $G'$ we maintain a feasible coloring with the same number of novel colors if for each $e \in E(G')$ with $\chi(e)$ not novel we assign $\chi(e) = c$.
Note that in particular, now all needles are colored $c$.

We define a \emph{pseudo needle} to be a triangle with at most one novel color where exactly two of the three vertices are incident to a needle and both edges incident to the non-needle have the same color.
To find a pseudo needle, we analyze the degrees within $C$.
Note that the sum of degrees in $C$ is $6m$ since each triangle has $3$ edges and thus contributes $6$ to the total degree. 
The cycles of $C$ form a tree in the following sense. We associate each triangle with a vertex and each vertex contained in two triangles with an edge. Recall that there are no vertices contained in more than two triangles.
Therefore the number of vertices contained in two triangles is exactly $m-1$.
The degree of each of these vertices is $4$. Each other vertex is incident to a needle. 
If a vertex is incident to a needle, it is incident to two edges within $C$.
Thus there are $(6m - 4(m-1))/2 > m$ vertices in $C$ incident to a needle. 
By the pigeonhole principle, there is a triangle $T_i$ incident to at least $2$ needles.

Note that in $T_i$ there cannot be two distinct novel colors incident to a needle as the common vertex would see at least three colors.
In particular, this implies that if $T_i$ is incident to three needles (and thus the only triangle in $C$), there is at most one novel color in $C$, a contradiction to having at least $m+1$ colors.
Otherwise let $\{u,v,w\}$ be the vertices of $T_i$, with $u$ and $v$ incident to a needle and $w$ not incident to a needle.
If $T_i$ is a pseudo needle, we remove $T_i$ and its two needles from $G'$ and introduce a new vertex $w'$ and a needle $ww'$ with $\chi(ww') = c$.
Then the resulting cactus has $m-1$ triangles and at least novel $m$ colors, not counting the possible novel color of $T_m$ -- a contradiction to $m$ being minimal.
We show that we can always turn $T_i$ into a pseudo needle, which then implies that $C$ and $\chi$ cannot exist.

If $T_i$ does not have a novel color, it is a pseudo needle. 
If $T_i$ has one novel color $c'$, we distinguish two cases.
If $w$ sees color $c$, we set $\chi(uv) = c'$ and $\chi(uw) = \chi(vw) = c$, obtaining a feasible coloring with the same number of colors and a pseudo needle.
Otherwise $\chi(uw) = \chi(vw) = c'$, and thus again $T_i$ is a pseudo needle.

If $T_i$ has two novel colors $c',c''$, both must be incident to $w$ as otherwise $u$ or $v$ would see three colors.
Let $w',w'' \notin \{u,v\}$ be the other two vertices adjacent to $w$.
Then $\chi(ww'),\chi(ww'') \in \{c'c''\}$; w.l.o.g.\ $\chi(ww') = c''$.
We change $\chi$ by setting $\chi(uw) = \chi(vw) = \chi(uv) = c'$.
We thus obtain a feasible coloring with the same number of colors but only one novel color in $T_i$, i.e., $T_i$ is a pseudo needle. 
 \end{proof}

\paragraph*{Normalized Graph.}

Some of our claims rely on the problem instance $G$ being a graph, such that none of Modifications~1--3 can be applied on $G$ anymore.
We define such a graph $G$ to be \textit{normalized}.
Below we show that one can efficiently compute a normalized graph $G'$ for every problem instance $G$, and that we can use the notion of normalized graphs to design approximation algorithms for the maximum edge $2$-coloring problem.

\begin{lemma}\label{lem:normalization}
  Given a graph $G$, in polynomial time we can compute a normalized graph $G'$ such that $\pair{G, G'}$ is an equivalent pair.
\end{lemma}
\begin{proof}
\Cref{lem:mod1,lem:mod2,lem:mod3} describe how Modifications 1--3 are performed, and it is shown that they can be applied in polynomial time when given a leaf, a degree-2 vertex, or a simple cactus respectively.
Furthermore, it is clear that if $G$ contains a leaf or a degree-2 vertex, we can find it efficiently.
Simple cacti can also be found in polynomial time as follows.

Let $\mathcal T$ be the set of those triangles in $G$ that only contain vertices of degree $3$ or $4$.
Observe that $\mathcal T$ contains all triangles that are part of simple cacti in $G$.

We say two triangles $T_1, T_2 \in \mathcal T$ are \emph{compatible} if they intersect at exactly one vertex $v$.
Observe that $v$ then must be a degree-4 vertex.
This means that if we fix a triangle $T_1 \in \mathcal T$ and a vertex $v \in V(T_1)$ there is at most one compatible triangle $T_2 \in \mathcal T$ intersecting at $v$.
If such a $T_2$ exists, any simple cactus which contains~$T_1$ must also contain $T_2$. 
If on the other hand, the degree of $v$ is $4$ and no such~$T_2$ exists, we may discard the triangle $T_1$, since it cannot be part of any simple cactus. Let $\mathcal T'$ be the set of remaining triangles.

If we want to determine, whether a given triangle $T \in \mathcal T'$ belongs to a simple cactus, we start with the set $\mathcal S = \{T\}$, which represents the partial cactus containing only $T$, and grow this set until either the triangles of $\mathcal S$ form a valid simple cactus, or until the process fails.
Note that each degree-$4$ vertex in a triangle of $\mathcal S$ joins that triangle with some compatible triangle from $\mathcal T$.
While there is a triangle $T_i$ with a degree-$4$ vertex in $\mathcal S$, such that the respective compatible triangle $T_j$ is not in $\mathcal S$, we distinguish two cases: 
If $T_i$ is the only triangle in $\mathcal S$ that is compatible with $T_j$, we add $T_j$ to $\mathcal S$ and continue the process.
If there is another triangle $T_k \in \mathcal S \setminus \{T_i\}$ which also is compatible to $T_j$, note that adding $T_j$ to $\mathcal S$ would create a cycle of length at least $4$ in the cactus, making it not a simple one.
However, since~$T_i$ and $T_j$ are compatible, any simple cactus containing $T_i$ must also contain $T_j$. Putting these two observations together, we get that $\mathcal S$ cannot be extended to a simple cactus; thus the process must fail.
Since there are $O(n^3)$ triangles, the whole process can be done in polynomial time.

We now complete the proof by showing that we can normalize a graph by applying a sequence of Modifications 1--3 that ends after polynomially many steps.
Observe that by applying Modifications 1 and 3, we will reduce the number of edges in $G$, but may create a polynomial amount of degree-$2$ vertices.
However, applying Modification~2 reduces the number of degree-$2$ vertices by one while keeping the number of edges the same.
Hence by applying Modification~2 first whenever it is possible, we reach a point when no further modifications can be made after polynomially many steps.
 \end{proof}

Let $C$ be a component of $G$ such that $C$ is an isolated vertex or $C$ has two vertices connected by an edge.
Then we say that $C$ is a \emph{trivial component}.
Otherwise, the component is called \emph{non-trivial}.
\begin{lemma}\label{lem:components}
  Suppose there is an $\alpha$-approximation algorithm \textsf{A} for each non-trivial component of a normalized graph $G'$ such that $\pair{G, G'}$ is an equivalent pair.
  Then there is an $\alpha$-approximation algorithm for~$G$.
\end{lemma}

\begin{proof}
    The statement follows by a simple application of the local ratio method.
    For the non-trivial components, we color the edges of graph $G'$ according to the output of \textsf{A}.
    \Cref{alg:basic} colors all trivial components optimally. Thus for each component we obtain an $\alpha$-approximation (or better)
    and the overall approximation ratio cannot be worse than $\alpha$.
 \end{proof}

Due to \Cref{lem:normalization,lem:components}, from now on we can assume that the problem instance is a normalized connected graph $G$ with more than two vertices. We now have a closer look at leaves. 
\begin{lemma}\label{lem:pendant_unique}\label{lem:optimal_pendant_unique}
Let $G$ be a graph where no two leaves share a neighbor and let $\chi$ be an edge 2-coloring of $G$.
Given $G$ and $\chi$, we can efficiently compute an edge 2-coloring $\hat\chi$ of $G$ that uses at least as many colors as $\chi$ and assigns each pendant edge of $G$ a unique color.
In particular, in a normalized graph there is an optimal edge-$2$-coloring that assigns a unique color to each pendant edge.
\end{lemma}

\begin{proof}
    Let $\chi$ be a coloring that does not assign a unique color to every pendant edge.
    We describe a transformation that turns a $\chi$ into a new coloring $\hat\chi$ with $|\hat\chi| \geq |\chi|$, which assigns a unique color to one more pendant edge than $\chi$ does.
    The statement then follows by iteratively applying this transformation.

    Choose an arbitrary pendant vertex $v$ whose pendant edge does not have a unique color under $\chi$. 
    Let $u$ be the unique neighbor of $v$ in $G$ and let $c_1 = \chi(u v)$.
    If $u$ only sees color~$c_1$, we simply introduce a new unique color for $u v$; observe that the new coloring is feasible.
    Otherwise, denote by $c_2$ the second color seen by $u$.  
    Since $v$ is the only leaf adjacent to $u$, neither $c_1$ nor $c_2$ are the unique color of a pendant edge.

    Let $\bar\chi$ be the coloring that arises from $\chi$ by assigning all edges colored with $c_2$ the color~$c_1$ as their new color. 
    Then $\bar \chi$ is valid and uses exactly one color less than $\chi$. 

    In $\bar\chi$, vertex $u$ only sees one color, which allows us to recolor the edge $u v$ with color $c_2$ in~$\hat\chi$.
    The resulting coloring $\hat\chi$ uses one more color than $\bar\chi$, hence the same number of colors as $\chi$.
    The pendant edge $u v$ has a unique color, and all pendant edges which had a unique color under $\chi$ still do so under~$\hat\chi$.
    Since in a normalized graph leaves are unique, also the last claim of the lemma follows.
 \end{proof}

\section{An upper bound on the optimal solution}\label{sec:upper-bound}

To show \Cref{lem:maxcolors}, we analyze the character graph of the given instance. 

\subsection{Preparing the character graph}\label{sec:character}
Intuitively, a character graph is an edge-induced subgraph with exactly one representative edge for each color.

\begin{definition}[Character Graph]
  Given an optimal solution $\chi$ for a graph $G$, a \emph{character graph} of $(G,\chi)$ is a subgraph $H$ with vertex set $V(G)$ 
  and coloring $\chi_{|E(H)}$  such that (i) for each $e,f \in E(H)$ with $e \neq f$, $\chi(e) \neq \chi(f)$ and (ii) for each edge $e \in E(G)$ there is an edge $f \in E(H)$ with $\chi(e) = \chi(f)$.
\end{definition}

For ease of notation, in the following we write $\chi$ instead of $\chi_{|E(H)}$.
Observe that in a character graph $H$, no vertex can have a degree larger than $2$ since otherwise there would be a vertex with three incident colors.
Thus $H$ is a collection of isolated vertices, paths and cycles.
We call a vertex in $H$ a \emph{free vertex} if its degree is zero, an \emph{end vertex} if its degree is one and an \emph{inner vertex} if its degree is two. 
We frequently use the following known simple but powerful lemma.
\begin{lemma}[Feng et al.~\cite{FengZW09}]\label{lem:character_subgraph}
    Let $\chi$ be a feasible 2-edge coloring of a graph $G$ and let $u \neq v$ be two vertices in $V(G)$.
    If $|\chi(u) \cup \chi(v)| \ge 4$, $u$ and $v$ are not adjacent in $G$.
    In particular, if $H$ is a character graph of $(G,\chi)$ and $u \neq v$ are two inner vertices that are not neighbors in $H$ 
    then  $uv \notin E(G)$.
\end{lemma}

Based on Lemma~\ref{lem:character_subgraph}, we can avoid cycles within a character graph.

\begin{lemma}\label{lem:cycle-free}
  Let $G$ be a normalized graph, and $\chi$ a coloring of $G$.
  Then there is a character graph $H$ of $(G,\chi)$ such that $H$ is cycle-free.
\end{lemma}
\begin{proof}
Let $H$ be a character graph of $G$. 
By definition, each edge of $H$ has a unique color. 
We claim that each connected component of a character graph can be converted to a path.
Suppose $H$ has a cycle $C$.
Since the graph is normalized, each vertex in $C$ has degree at least three. 
According to \cref{lem:character_subgraph}, the cycle $C$ cannot contain a chord in $G$, since a chord would be incident to two inner vertices of $H$. 
Thus, we consider a vertex $v\notin C$ which is adjacent to a vertex $u \in C$.
Let us denote the two neighbors of vertex $u$ in $C$, by $w_1$ and $w_2$. 
By definition, the edges $w_1 u$ and $w_2 u$ have different colors $c_1 \neq c_2$ and therefore $\chi(v u) \in \{c_1,c_2\}$. 
W.l.o.g, we assume $\chi(v u) = c_1$.
We replace $w_1 u$ by $v u$ in $H$ and we obtain a new character graph $H'$, which does not contain the cycle $C$. 
We iteratively apply this change for all cycles, until $H'$ becomes cycle-free and thus the lemma follows.
 \end{proof}

To further structure the character graph, we introduce a reachability measure.

\begin{definition}\label{def:scope}
    Let $v$ be a vertex of a character graph $H$ of $(G,\chi)$.
    The \emph{scope} of $v$ ($\scope(v)$) is the set of vertices defined inductively as follows within $G$:
    \begin{enumerate}[(i)]
        \item $v \in \scope(v)$.
        \item If $u \in \scope(v)$ and there is an edge $e = u u' \in E(H)$ for an inner vertex $u'$,
            then $V(\chi(e)) \subseteq \scope(v)$ (i.e., we include the color class of $\chi(e)$).\label{scope:color}
    \end{enumerate}
    We may choose a total ordering of the vertices and in \textsf{(\ref{scope:color})}, we always choose the smallest vertex that satisfies the properties.
    Let $\kappa \ge 0$ be an integer which is at most the number of color classes added and let $c_i$ be the color of the $i$-th color class added, for $1 \le i \le \kappa$.
    The \emph{scope graph} of $v$ and $\kappa$ for a given ordering is the graph $(\scope(v),F)$, where $F:= \bigcup_{i=1}^\kappa E(c_i)$. 
    We skip the ordering and say that subgraph of $G$ is a scope graph of $v$ and $\kappa$ if there exists an ordering for which it is a scope graph of $v$ and $\kappa$.
\end{definition}
Note that the (total) scope of a vertex does not depend on the chosen ordering.
The scope of a vertex captures a natural sequence of dependencies between edge colors.
In the following lemma, we show how to avoid free vertices in the scope of vertices, a property which is important in the proof of \Cref{lem:maxcolors}.

\begin{lemma}\label{lem:singlescope}
    Let $v$ be an inner vertex such that $vv' \in E(H)$ for an inner vertex $v'$.
    Each character graph $H$ of $(G,\chi)$ can be transformed into a character graph~$H'$ of $(G,\chi)$ such that there is no free vertex in $\text{scope}(v)$.
\end{lemma}
\begin{proof}

    Suppose $H$ does not have the described properties. 
    Then there is an inner vertex~$v$ that has a free vertex in its scope.
    Let $\kappa$ be the number of color classes added to the scope graph of $v$ until the first free vertex is reached.
    We show by induction on $\kappa$ that we can reduce the number of free vertices.
    By iteratively applying the argument, the lemma follows since there are at most $n$ free vertices in $H$.

    Note that $\{v,v'\} \subseteq \scope(v)$ and if $\scope(v) = \{v,v'\}$ there is trivially no free vertex in $\scope(v)$.
    Thus $\kappa \ge 1$ and let $C$ be the $\kappa$th color class added to the scope graph of $v$. 
    
    Let $c$ be the color of $C$, $w$ the first free vertex reached and
    let $e = u u'$ be the edge of $H$ of color $c$, according to \Cref{def:scope}.

    Thus $e \in E(H)$ and
    there is an edge $f = w w'$ in the scope graph of $v,\kappa$ with $\chi(f) = c$.
    We add $f$ to $H$ and remove $e$ from $H$, which leads to a new valid character graph.
    Observe that $w'$ can only be an inner vertex if $w' \in \{u,u'\}$, by \Cref{lem:character_subgraph}. 
    In addition, $w'$ is either an end vertex or another free vertex.
    We have therefore simply extended a path in $H$ or created a new path.
    (In particular, $f$ cannot close a cycle since $w$ is a free vertex.)

    By construction, $u'$ is an inner vertex. Removing $e$ therefore turns $u'$ into an end vertex.
    
    If $u$ is an inner vertex (before modifying $H$), it also becomes an end vertex and we have reduced the number of free vertices.
    Otherwise, $u$ is an end vertex and becomes a free vertex. 
    However, in that case $u$ is contained in a color class $C'$ created earlier in $\scope(v)$ according to \Cref{def:scope}.

    In both cases we therefore make progress: we either directly reduce the number of free vertices or we ``move'' the free vertex closer to $v$.
    Eventually, if the considered color class is that determined by $v$, we use that $v$ is an inner vertex and thus cannot be turned into a free vertex.
 \end{proof}

We further extend the notion of scope to a set of vertices, which is not merely the union of scopes. 
\begin{definition}\label{def:setscope}
    Let $S$ be a set of vertices of a character graph $H$ of $(G,\chi)$.
    If $S = \{v\}$, we define the scope as $\scope(v)$ and define the scope graph accordingly.
    For $|S| > 1$, the scope of $S$ ($\scope(S)$) is the set of vertices defined inductively as follows within $G$:
    \begin{enumerate}[(i)]
        \item $\scope(v) \subseteq \scope(S)$ for each $v \in S$.
        \item \label{scopeset:one} Let $uu' \in E(H)$ be the only edge of a path in $H$. 
             Let $\kappa,\kappa'$ be numbers and $U,U' \subset S$ sets with $U \cap U' = \emptyset$.
            If $u$ is in a scope graph of $U$ for $\kappa$ and $u'$ is in a scope graph of $U'$ for $\kappa'$ such that the colors of the two graphs are disjoint, then $V(\chi(uu')) \subseteq \scope(U \cup U')$.
        \item If $u \in \scope(U)$ for $U \subseteq S$ and there is an edge $e = u u' \in E(H)$ for an inner vertex $u'$,
            then $V(\chi(e)) \subseteq \scope(U)$. \label{scopeset:color}
    \end{enumerate}
    If $T$ is a scope graph for $U$ and $\kappa$ and the color class of $c$ is added, then
    $(V(T) \cup V(\chi(c)),E(T) \cup E(\chi(c)))$ is a scope graph for $U$ and $\kappa+1$.
    If additionally $T'$ is a scope graph of $U'$ and \textsf{(\ref{scopeset:one})} applies with respect to the two scope graphs,
    $((V(T) \cup V(T') \cup V(\chi(uu')),E(T) \cup E(T') \cup E(\chi(uu')))$ is a scope graph for $U \cup U'$ and $\kappa + \kappa' + 1$.
\end{definition}
We now strengthen \Cref{lem:singlescope}.
\begin{lemma}\label{lem:scope}
    Let $S$ be a set of vertices such that each $v \in S$ is an inner vertex such that $vv' \in E(H)$ for an inner vertex $v'$.
    Each character graph $H$ of $(G,\chi)$ can be transformed into a character graph~$H'$ of $(G,\chi)$ such that there is no free vertex in $\text{scope}(S)$.
\end{lemma}
\begin{proof}
We only have to handle condition \textsf{(\ref{scopeset:one})}. The remaining proof is analogous to \Cref{lem:singlescope}.
Let $c := \chi(uu')$ and suppose that $w$ is a free vertex incident to an edge $e$ colored $c$.
We then add $e$ to $E(H)$ and remove $uu'$.
The removal creates two free vertices.
However, since $U \cap U' = \emptyset$, we can recursively move the free vertices to previous color classes.
If $|U| = 1$ or $|U'| = 1$, we apply \Cref{lem:singlescope}.
 \end{proof}

We say that a character graph is \emph{nice} if it is (i) cycle-free and (ii) there is no free vertex in  $\scope(S)$ for an arbitrary set $S \subseteq V(G)$ such that $vv' \in E(H)$ for $v \in S$ and an inner vertex $v'$.
The following lemma gives some guarantees for the existence of such a character subgraph.
  
\begin{lemma}
  Each normalized connected graph $G$ with optimal coloring $\chi$ has a nice character graph $H$.
\end{lemma}
\begin{proof}
The lemma follows directly from \Cref{lem:cycle-free} and \Cref{lem:scope}, noting that in the proof of \Cref{lem:scope} we do not introduce new cycles. 
Furthermore, given a normalized graph $G$ and the coloring provided by \Cref{lem:pendant_unique}, 
we may assume that $H$ is a nice character graph where all pendant vertices of $G$ are endpoints of paths in $H$.
 \end{proof}

\subsection{The proof of \Cref{lem:maxcolors}}

With the preparation of \Cref{sec:character}, in this section we prove our main lemma.
\maxcolors*

Let $G$ be a normalized problem instance and $H$ a nice character graph of $(G,\chi)$.
Let~$F$ be the set of free vertices, $T$ the set of end vertices and $I$ the set of inner vertices of~$H$.
We define $n := |V(H)|$, $i := |I|$, $f := |F|$ and $t := |T|$, and clearly $n = i+t+f$.
We show that there is a mapping $\iota$ from the set $I$ of inner vertices to $T \cup F$ with the property that for each vertex $v$ from $T \cup F$ there is at most one inner vertex mapped to $v$ if $v \in T$ and at most two inner vertices are mapped to $v$ if~$v \in F$.
Intuitively, we can see $\iota$ as an injective mapping, where each vertex in $F$ is split into two vertices.
The reason is that we can see a free vertex as a path of length zero and we count two end vertices for each path.
Furthermore, we maintain that $\iota$ never maps an inner vertex to a pendant vertex of $G$.
We first show that the mapping implies \Cref{lem:maxcolors}.
\begin{lemma}\label{lem:if-iota}
    Suppose $\iota$ exists. Then no feasible coloring has more than $3n/4 - \ell/4$ colors.
\end{lemma}
\begin{proof}
    The mapping $\iota$ implies $2f + t - \ell \geq i = n-t-f$ and thus $3f + 2t \ge n + \ell$
    since $\ell$ out of $t$ end vertices cannot be used for the assignment.
    Each free vertex and each path is a component of the character graph $H$. 
    Therefore the number of components is $f + t/2$, which is minimized if $f = 0$ and $t = (n+\ell)/2$. 
    Thus there are at least $\frac{n+\ell}{4}$ components in $H$ and
    the number of colors is at most $n - n/4 - \ell/4 = 3n/4 - \ell/4$, completing the proof.
 \end{proof}
  
We now construct the mapping $\iota$.
We associate the vertices with distinct natural numbers $\{1,2,\dotsc,n\}$ and define $\iota$ iteratively.
An end vertex is \emph{saturated} if there is an inner vertex mapped to it and a free vertex is saturated if there are \emph{two} inner vertices mapped to it.
While there are unassigned inner vertices, we continue the following process.
Let $v \in I$ be the unassigned vertex with the smallest index.
We define the set
$U_v := \{u \in F \mid uv \in E(G) \setminus E(H) \text{ and } |\iota^{-1}(u)| \leq 1\} \cup  \{u \in T \mid uv \in E(G) \setminus E(H) \text{ and } \iota^{-1}(u) = \emptyset \}$,
that is, the set of unsaturated free- and end vertices adjacent to $v$ via an edge \emph{outside} of~$H$.
We remark that $v$ and $u$ are allowed to be in the same path of $H$, as long as they are not adjacent in~$H$.
If $U_v \neq \emptyset$, we set $\iota(v) := \min_{u \in U_v} u$. 
Clearly, if $U_v \neq\emptyset$, we find a valid mapping for~$v$ and can continue. 
If $U_v = \emptyset$, we add $v$ to a set $Q$ of postponed vertices and continue with the next vertex.

To finish the construction, we have to map the postponed vertices.
Recall that if $v$ is a vertex in~$Q$, then $U_v = \emptyset$ holds.
We then find a vertex $u$ to map $v$ to by growing a \textit{plain} cactus: a plain cactus is a triangular cactus without needles that is a degree-$4$ bounded subgraph $C$ of the problem instance $G$ where each vertex of $C$ that is connected to $G \setminus C$ can have any number of adjacent vertices in $V(G) \setminus V(C)$, as opposed to  \textit{one} in case of a \textit{simple} cactus.
In particular, we will be growing a plain cactus which is not a simple cactus.
As a simple cactus, a plain cactus can have vertices adjacent in $G$. We call the edge between these vertices a \emph{cactus chord}. 

To gain some intuition, we first argue how to grow an initial triangle of the cactus (see also \Cref{fig:scope_cactus}).
There are no degree-two vertices in $G$, therefore $v$ has a neighbor $u'$ in $G$ that is not a neighbor of $v$ in $H$.
Due to \Cref{lem:character_subgraph}, $u'$ is not an inner vertex.
We note that $u'$ cannot be a free vertex either: 
if $u'$ was a free vertex and we could not map $v$ to~$u'$, then~$u'$ would be saturated and $v$ would be the third inner vertex adjacent to $u'$.
By \Cref{lem:scope}, however, that would mean $u'$ seeing 3 colors, as each of the three inner vertices would be connected to $u'$ with different colors, contradicting the feasibility of $\chi$.
Hence $u'$ is an end vertex.
    
Since we cannot map $v$ to~$u'$, there must be another vertex $\hat{v}$ already mapped to $u'$.
Let $P$ denote the path of $u'$ and let~$u''$ be the vertex adjacent to $u'$ in $P$.
Observe that $u'' \notin \{v,\hat{v}\}$ because otherwise the edge $u'v$ or $u'\hat{v}$ would be contained in $E(H)$, but both have to be in $E(G)\setminus E(H)$ in order to be considered to map to $v$ or $\hat{v}$, respectively.
Since $v$ and $\hat{v}$ are inner vertices, the colors of $u'v$ and $u'\hat{v}$ have to be from $\chi(v)$ and $\chi(\hat{v})$, respectively.

By \Cref{lem:character_subgraph}, $\chi(u'v) \neq \chi(u'u'')$ and $\chi(u'\hat{v}) \neq \chi(u'u'')$.
Therefore $c := \chi(u'v) = \chi(u'\hat{v})$ and, again by \Cref{lem:character_subgraph}, $v$ has an incident edge $e$ and $\hat{v}$ has an incident edge~$\hat{e}$ in $H$ with $\chi(e) = \chi(\hat{e}) = c$, which implies $e = \hat{e}$.
Therefore, $v$ and $\hat{v}$ are neighbors in the same path $P'$ of $H$; note that $P'$ is not necessarily distinct from $P$.

\begin{figure}[tb]
  \centering
  \includegraphics[scale=0.39]{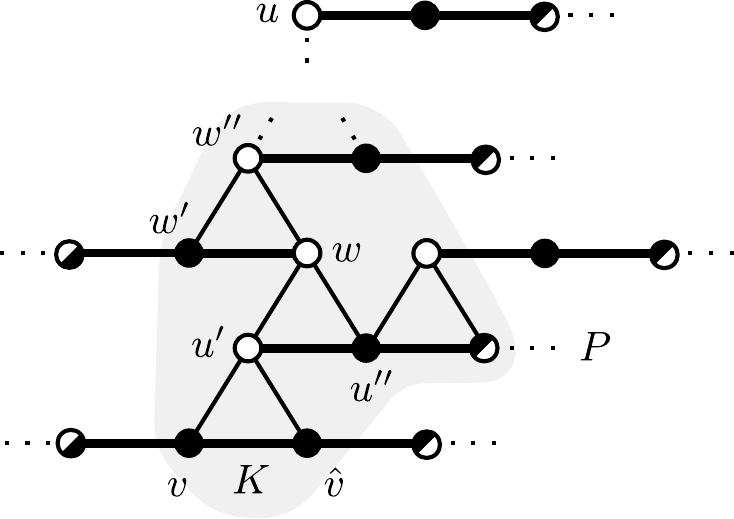}
  \caption{End vertices are marked by hollow circles, inner vertices are marked by filled circles.
  A half-filled vertex can be either an inner- or an end vertex.\label{fig:scope_cactus}}
\end{figure}

We observe that if there are only the edges from $H$ and the edges from the triangle formed by $v,\hat{v},u'$ incident to these three vertices, 
we have a simple cactus, i.e., a cactus of the form removed by Modification~3.
Therefore there is another vertex adjacent to~$v$,~$\hat{v}$, or~$u'$.
We now argue that no matter which vertex has another adjacent vertex, we can either extend the plain triangular cactus or find the aimed-for vertex $u$. 

More precisely, starting from $V(K) := \{v\}$ 
we grow a set of vertices $V(K)$; this process is shown in detail as \Cref{alg:mapping}.
Formally, we also grow an edge set $E(K)$ such that $K$ is the aimed-for cactus. 
For a cactus $K$ let $N(K) := \{\tilde{v} \in V(G) \setminus V(K) \mid \text{there is a vertex } \tilde{u} \in V(K) \text{ with }\tilde{u}\tilde{v} \in E(G)\}$.
\begin{algorithm}
    \caption{Mapping a vertex $v \in Q$.\label{alg:mapping}}
    \begin{algorithmic}[1]
        \State Let $V(K) := \{v\}$ and $E(K) := \emptyset$;
        \While{$\iota(v)$ is not yet determined}
        \If{$N(K)$ contains a not saturated free vertex or end vertex $u$}
        \State $\iota(v) := u$;
        \Else
        \State Find $w \in V(K)$ and $w',w'' \in V(G) \setminus V(K)$ 
            with $w' \neq w''$, $ww' \in E(H), ww'' \in E(G)$;
        \State $V(K) := V(K) \cup \{w',w''\}$ and $E(K) := E(K) \cup \{ww',w'w'',w''w\}$;
        \If{$w' \in V(K')$ for a previously considered cactus $K'$}
        \State $V(K) := V(K) \cup V(K') \cup \{w''\}$ and $E(K) := E(K) \cup E(K') \cup \{ww',ww'',w'w''\}$;
        \EndIf
        \EndIf
        \EndWhile
    \end{algorithmic}
\end{algorithm}

We will show that the cactus constructed by \Cref{alg:mapping} satisfies the following invariants.
\begin{enumerate}
    \item The graph $K$ is a plain cactus. \label{inv:cactus}
    \item Each inner vertex of $K$ except for $v$ is mapped to an end vertex within $V(K)$. \label{inv:closed}
    \item For each end vertex of $K$, there is an inner vertex in $V(K)$ mapped to it. \label{inv:symmetric}
    \item All triangles of $K$ are monochromatic.\label{inv:monochromatic}
    \item For each color $c$ in $K$, there is a color-$c$ edge $e$ in $K$ that is also in $E(H)$. \label{inv:H}
    \item $V(K) \subseteq \text{scope}(Q)$.\label{inv:scope}
    \item $V(K)$ does not include free vertices. \label{inv:no_free}
    \item Each degree-$2$ vertex of $K$ is incident to an edge from $E(H) \setminus E(K)$. \label{inv:degree-2}
\end{enumerate}

We first show the invariants assuming that $V(K) \cup N(K)$ does not contain vertices from previously constructed cacti.
Observe that the initial cactus with $V(K) = \{v\}$ satisfies all invariants.
From now on we assume that $K$ is a cactus constructed during the execution of \Cref{alg:mapping} and $K$ satisfies all invariants. 

Our proof is based on the following technical lemmas.
 
\begin{lemma}\label{lem:two-incident-colors}
    Let $u$ be an inner or a saturated vertex.    
    Then $|\chi(u)| = 2$.
    Furthermore, if $u \notin V(K)$ and it is either an inner vertex or $\iota^{-1}(u) \cap V(K) = \emptyset$,
    then $\chi(u) \cap \chi(K) = \emptyset$.
\end{lemma}
\begin{proof}
    If $u$ is an inner vertex, the claim is true by definition of inner vertices and Invariant~\ref{inv:H} 
    If $u$ is an end vertex, it is incident to an edge $e \in E(H)$, which is not in $E(K)$ if $u \notin V(K)$.
    Let $\{u'\} := \iota^{-1}(u)$.
    Since $e$ is not incident to $u'$ and $u'$ sees two colors, $|\chi(u)| = 2$.
    By Invariant~\ref{inv:H}, none of these colors are contained in $\chi(K)$.
    If $u$ is a free vertex, there are two vertices $u' \neq u''$ with $\iota^{-1}(u) = \{u',u''\}$.
    Due to Lemma~\ref{lem:scope}, $u'$ and $u''$ are not adjacent in $H$.
    Then the claim follows by applying the argument of the end-vertex case twice.
 \end{proof}

    The next lemma shows that growing a cactus is independent of vertices that connect two triangles.
\begin{lemma}\label{lem:degree4}
    Let $\tilde{u}$ be a degree-$4$ vertex of the cactus $K$. Then $\tilde{u}$ is also a degree-$4$ vertex in $G$.
\end{lemma}
\begin{proof}
    By contradiction, suppose that the degree of $\tilde{u}$ is at least $5$ in $G$.
    Then $\tilde{u}$ is incident to two triangles of $K$ and it sees the two colors of these triangles.
    Let $u'$ be a fifth adjacent vertex.\\
    Case 1: Suppose $u' \notin V(K)$.
    Then $\tilde{u}u' \notin E(H)$ since both colors in $\chi(\tilde{u})$ have their edges from $H$ within the cactus, by Invariant~\ref{inv:H}.
    The same invariant and \Cref{lem:character_subgraph} imply that $u'$ is not an inner vertex.  
    Among the two triangles incident to $\tilde{u}$, the one added later to $V(K)$ was only added if $\tilde{u}$ cannot be mapped to $u'$.
    Due to Invariants \ref{inv:closed} and \ref{inv:symmetric}, previous mappings to $u'$ do not originate in $V(K)$. 
    
    There are two possible situations that can occur: $u'$ is either a saturated free vertex or a saturated end vertex.
    In both cases, it sees two distinct colors by \Cref{lem:two-incident-colors}.
    Note that it does not matter if $u'$ was mapped directly or using \Cref{alg:mapping}, due to Invariants~\ref{inv:monochromatic}, \ref{inv:H} and \ref{inv:degree-2}.
    In both cases $|\chi(\tilde{u}) \cup \chi(u')| = 4$ and thus there is no valid color for $\tilde{u}u'$ by \Cref{lem:character_subgraph}.

    Case 2: Suppose $u' \in V(K)$. Then $u'$ cannot be an inner vertex or a degree-$4$ vertex, by \Cref{lem:character_subgraph}.
    Since by Invariant~\ref{inv:no_free} $V(K)$ does not contain free vertices, it has to be an end vertex of degree $2$ in $V(K)$.
    However, by Invariants~\ref{inv:monochromatic}, \ref{inv:H} and \ref{inv:degree-2}, $u'$ then also sees two colors distinct from $\chi(\tilde{u})$. 
 \end{proof}

Furthermore, cactus chords do not influence the process.
\begin{lemma}\label{lem:degree3}
    Suppose $K$ has two vertices $u',u'' \in V(K)$ such that $u'u'' \notin E(K)$ but $u'u'' \in E(G)$, and $N(K)$ does not have unsaturated vertices.
    Then in $G$, the degrees of both $u'$ and $u''$ are three.
\end{lemma}
\begin{proof}
    By \Cref{lem:degree4}, the degrees of both vertices in $V(K)$ are two.
    We first argue that $u'$ sees two colors.
    If $u'$ is an internal vertex this is true by definition.
    Otherwise $u'$ is an end vertex and therefore incident to one edge from $H$.
    By Invariant~\ref{inv:degree-2}, the edge is not in $E(K)$ and thus the claim follows from Invariants~\ref{inv:monochromatic} and \ref{inv:H}.
    Symmetrically, the same claim holds for $u''$.    
    In particular, $u'u'' \in E(H)$ since otherwise $|\chi(u') \cup \chi(u'')| = 4$ and \Cref{lem:character_subgraph} applies. 

    The vertex $u'$ does not have further neighbors since otherwise by \Cref{lem:two-incident-colors}, it would see more than two colors.
    The claim for $u''$ follows symmetrically.
 \end{proof}

With this preparation we show that all vertices and edges within \Cref{alg:mapping} exist.
\begin{lemma}\label{lem:w_exists}
    In \Cref{alg:mapping}, the vertices $w,w',w''$ exist, $\iota(w') = w''$, $\chi(ww') = \chi(ww'') = \chi(w'w'')$, and there is no other vertex $\bar{w} \in V \setminus V(K)$ adjacent to $w$.
\end{lemma}
\begin{proof}
    Since $K$ is a plain cactus and we cannot apply Modification~$3$, there have to be edges not allowed in a simple cactus. These cannot be degree-$4$ vertices, due to \Cref{lem:degree4}.
    Also, by \Cref{lem:degree3}, edges between vertices from $K$ cannot violate the properties of simple cacti.
    Thus there has to be a degree-$2$ vertex $w$ of $K$ which has a set $A$ of at least two adjacent vertices in $V(G) \setminus V(K)$.
    By Invariant~\ref{inv:degree-2}, there is a vertex $w' \in A$ with $ww' \in E(H)$.
    Let $c$ be the color of the triangle containing $w$ and $c' := \chi(ww')$.

    We claim that there is no vertex $\hat{w} \in A$ with $\chi(w\hat{w}) \neq c'$. Since $w$ only sees two colors, the only possibility would be $\chi(w\hat{w}) = c$.
    However, since $w$ sees two colors, $\hat{w}$ cannot be an inner vertex and by \Cref{lem:two-incident-colors} it also cannot be a free vertex or an end vertex.

    Next, we show that $\iota(w') = w''$. 
    Suppose $\iota(w') \neq w''$.
    Then, by \Cref{lem:two-incident-colors},
    it sees two colors which are distinct from $c$, and also distinct from $c'$ since $ww' \in E(H)$ and $w,w' \notin \iota^{-1}(w'')$.
    However, $w''$ also sees  $c'$ and therefore more than two colors, a contradiction. Thus $\iota(w') = w''$.
    In particular, $\chi(ww') = \chi(ww'') = \chi(w'w'')$.
 \end{proof}

With \Cref{lem:w_exists} we know that adding the triangle $\{w,w',w''\}$ does not violate the conditions of plain cacti, i.e., adding it still satisfies Invariant~\ref{inv:cactus}.
Since $\iota(w') = w''$, also Invariants~\ref{inv:closed} and \ref{inv:symmetric} are satisfied.
Invariant~\ref{inv:monochromatic} follows directly from \Cref{lem:w_exists}.
Since $ww' \in E(H)$ by definition (within \Cref{alg:mapping}), Invariant~\ref{inv:H} follows.
Invariant~\ref{inv:scope} follows by noting that $w \in \text{scope}(v)$, the edge $ww'$ satisfies the conditions of the lemma, and thus the color class $c$ is added to the scope.
Invariant~\ref{inv:no_free} is a direct consequence of Invariant~\ref{inv:scope}.
Invariant~\ref{inv:degree-2} requires additional arguments.

\begin{lemma}\label{lem:invariant-8}
    Invariant~\ref{inv:degree-2} is satisfied.
\end{lemma}

\begin{proof}
    Let $\tilde{u}$ be a degree-$2$ vertex. 
    We distinguish the type of $\tilde{u}$. 
    Due to Invariant~\ref{inv:no_free}, $\tilde{u}$ is not free. 
    If $\tilde{u}$ is an inner vertex, only one of the two incident edges from $E(H)$ can be in $E(K)$, due to Invariant~\ref{inv:monochromatic}. 
    Thus, due to the other edge, the invariant is satisfied.
    If $\tilde{u}$ is an end vertex, it has to have an incident edge $\tilde{u}u'$ not in $E(K)$, due to Modification~$2$.
    If $u' \in V(K)$, \Cref{lem:degree3} implies the invariant.
    We therefore may assume that $u' \notin V(K)$.

    By Invariant~\ref{inv:symmetric}, there is a vertex $\bar{u} \in V(K)$ with $\iota(\bar{u}) = \tilde{u}$. 
    In particular, $\bar{u}$ is an inner vertex, $\bar{u}\tilde{u} \in E(K)$ and $\bar{u}\tilde{u} \notin E(H)$. 
    Let $\hat{u}$ be the third vertex in the same triangle.
    By Invariant~\ref{inv:monochromatic}, $\chi(\hat{u}\tilde{u}) = \chi(\hat{u}\bar{u}) = \chi(\tilde{u}\bar{u})$.
    But then $\hat{u}\bar{u} \in E(H)$ as otherwise $\bar{u}$ would see three colors.
    Thus $\hat{u}\tilde{u} \notin E(H)$.
    However, since $u$ is an end vertex, it is incident to an edge from $E(H)$, which then can only lead outside of $K$.
 \end{proof}

Finally, we argue that all invariants are also preserved when merging two cacti.
For a cactus $K'$, let $v(K')$ be the vertex from $Q$ mapped using cactus $K'$.

\begin{lemma}\label{lem:join-cacti}
    Let $w',w''$ be the vertices from \Cref{alg:mapping} and let $V(K')$ be the vertex set of a cactus constructed in a previous application of \Cref{alg:mapping}.
    Then $w'' \notin V(K')$. 
    Furthermore, if $w' \in V(K')$, $v(K')$ is mapped to $w''$.
\end{lemma}
\begin{proof}
    If neither $w'$ nor $w''$ are in $V(K')$, there is nothing we have to show.
    Suppose $w' \notin V(K')$.
    Then also $w'' \notin V(K')$ since by \Cref{lem:w_exists}, $w'$ is mapped to $w''$ which by Invariant~\ref{inv:symmetric} would imply that both $w'$ and $w''$ are in $V(K')$.
    We therefore may assume $w' \in V(K')$.
    Now suppose that both $w'$ and $w''$ are in $V(K')$.
    Then there is a vertex $\tilde{w} \in V(K')$ such that $\{w',w'',\tilde{w}\}$ form a triangle within $K'$.
    By Invariant~\ref{inv:monochromatic}, the triangle is monochromatic and by Invariant~\ref{inv:H}, the edge in $H$ with color $\chi(w'w'')$ is contained in the triangle.
    However, this contradicts \Cref{lem:w_exists} according to which the edge in $H$ colored $\chi(w'w'')$ is contained in the triangle formed by $\{w,w',w''\}$.
    Thus $w'' \notin V(K'')$.

    Since we could not map $v$ to $w''$, by \Cref{lem:w_exists}, $w'$ is mapped to $w''$.
    If $w'$ is an inner vertex, we have a contradiction since by Invariant~\ref{inv:closed}, $w'$ would be mapped to a vertex within $V(K')$.
    Thus $w'$ is an end vertex, which implies that via $w'$, $v(K')$ is mapped to $w''$. 
 \end{proof}

Then composing a new cactus from $K, K'$, and the triangle formed by $\{w,w',w''\}$ satisfies all conditions:
Since by induction they are satisfied by $K$ and $K'$, we only have to check the new triangle $\{w,w',w''\}$.
We obtain a plain cactus since the degrees of $w$ and $w'$ are four and the degree of $w''$ is two.
The inner vertex $v(K')$ is mapped to $w''$ which implies Invariant~\ref{inv:closed} and \ref{inv:symmetric}.
The triangle $\{w,w',w''\}$ satisfies Invariant~\ref{inv:monochromatic} by \Cref{lem:w_exists}.

Since $ww' \in E(H)$, Invariant~\ref{inv:H} is satisfied, both $w$ and $w'$ are in the scope of $Q$ and $K,K'$ provide disjoint scope graphs, the conditions of \Cref{def:setscope} and therefore Invariant~\ref{inv:scope} are satisfied.
Invariant~\ref{inv:no_free} follows from Invariant~\ref{inv:scope} and that $Q$ satisfies the conditions of \Cref{lem:scope}.
Finally, Invariant~\ref{inv:degree-2} follows since $w''$ is an end vertex and its incident edge from $H$ is not in the constructed cactus.

\section{Subcubic Graphs and Claw-Free Graphs}\label{sec:results}

The upper bound shown in Section~\ref{sec:upper-bound} directly gives the following result.

\begin{lemma}\label{thm:normalized_perfect_matching}
  There is a {polynomial-time} $1.5$-approximation for normalized graphs that contain a perfect matching.
\end{lemma}
\begin{proof}
  Let $G$ be a normalized graph on $n$ vertices that has a perfect matching $M$.
  Recall that we can assume without loss of generality that $G$ is a connected graph with more than $2$ vertices.
  By \Cref{lem:maxcolors}, there is no feasible coloring of $G$ with more than $3n/4$ colors,
  and \Cref{alg:basic} obtains at least $|M| = n/2$ colors.
  Thus the attained approximation ratio is at most
    $\frac{3n/4}{n/2} \leq 3/2$.
 \end{proof}

In particular, our algorithm solves the tight worst-case instance depicted in \Cref{fig:tight_example_53} for the algorithm of Adamaszek and Popa~\cite{AdamaszekPopa2016} optimally.
\begin{figure}[tb]
  \centering
  \includegraphics[scale=0.39]{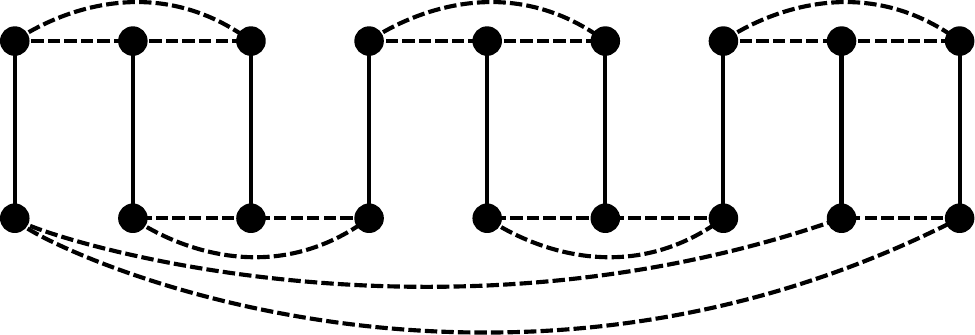}
  \caption{Worst-case instance containing a perfect matching from \cite{AdamaszekPopa2016}. 
      Without applying modifications, \Cref{alg:basic} may chose a perfect matching such that removing the matching leaves a connected graph -- resulting in a $5/3 \approx 1.667$ approximation. 
      An improved approximation ratio of $1.625$ is possible due to Modification~3.
    In particular, Modification~3 replaces all simple cacti by independent edges, which results in a modified graph $G'$ consisting only of independent edges, which means that \Cref{alg:basic} actually computes an optimal solution.
\label{fig:tight_example_53}}
\end{figure}
In order to handle subcubic graphs, we need to introduce additional modifications.
The validity of these modifications follows from \Cref{lem:pendant_unique}, which says that we can always assume that pendant edges are colored with a unique color.
\paragraph*{Modification~4: Bridge Removal}
The next modification is only applied to subcubic graphs, after none of Modifications~1--3 can be applied anymore.
Note that between two consecutive applications of Modification~4, we may have to update the graph by applying Modifications~1--3. 

Suppose we have such a normalized, subcubic graph $G$ with a cut of size one consisting of an edge $uv$, that is, the removal of $uv$ disconnects $G$ (see \Cref{fig:bridge_removal}).
Due to the definition of $G$, $u$ and $v$ are either degree-$1$ or degree-$3$ vertices: their degree cannot be more than $3$, and if it was $2$, Modification~2 could be applied.

Intuitively, Modification~4  
disconnects $G$ by removing all edges adjacent to $uv$ and, if they exist, connects the former neighbors of $u$ and $v$, respectively.
Formally, let $v,u_1,u_2$ be the neighbors of $u$ and $u,v_1,v_2$ the remaining neighbors of $v$, if they exist.
We remove the edges $u_1u$, $u_2u$, $v_1v$, and $v_2v$. Then we add the edge $u_1u_2$ if $u_1$ and $u_2$ exist and $v_1v_2$  if $v_1$ and $v_2$ exist.

\begin{figure}[tb]
  \centering
  \includegraphics[scale=0.39]{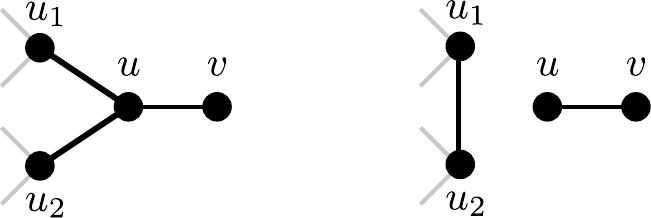}
  \caption{Modification~4: Bridge removal, when $\deg(u)=3$ and $\deg(v)=1$\label{fig:bridge_removal}.
  \label{fig:modification4}
  }
\end{figure}

\begin{lemma}\label{lem:mod4}
  Modification~4 is valid.
\end{lemma}
\begin{proof}
  We consider 3 cases.

  \parmacro{Case 1:} If $\deg(u)=\deg(v)=1$, Modification~4 has no effect.
  
  \parmacro{Case 2:} $\deg(u)=3$ and $\deg(v)=1$; the proof goes analogously if $\deg(u)=1$ and $\deg(v)=3$.
  Since we have a subcubic graph, $u$ has no neighbors besides $u_1$, $u_2$ and $v$.
  For every given optimal coloring $\chi^\star$ of $G$, we may assume that every pendant edge has a unique color, according to \Cref{lem:optimal_pendant_unique}.
  Since $uv$ has a unique color $c_1$, the edges $u_1 u$ and $u_2 u$ have to be colored with the same color $c_2 \neq c_1$.
  
  After the modification, the now disjoint edge $uv$ can keep the color $c_1$.
  If we introduced a new edge $u_1 u_2$, it can be colored by $c_2$, as this way all vertices see the same set of colors as before the modification.

  Therefore $|\chi'| \geq |\chi^\star$, hence $\opt(G) \leq \opt(G')$.
  Applying the steps in reverse, one can also construct a coloring $\chi$ in $G$, given a coloring $\chi'$ in $G'$, so that $|\chi| = |\chi'|$: the edges $u_1 u$ and $u_2 u$ can inherit the color of $u_1 u_2$, and then every vertex will see at most $2$ colors.
  
  It remains to deal with the case where there existed an edge $u_1 u_2$ in $G$.
  Performing Modification~4 on $uv$ then results in two copies of the edge $u_1u_2$.
  If both copies are colored~$c_2$, simply remove one copy.
  If one of the copies has a different color $c_3$, observe that there must be other edges colored by either $c_2$ or $c_3$ (or both) at $u_1$ or $u_2$.
  Remove a copy of $u_1 u_2$ that has such a color.
  The arguments about both directions follow.

  \parmacro{Case 3:} $\deg(u)=\deg(v)=3$.
  Given an optimal coloring $\chi^\star$ we can assume that the edge $uv$ has a unique color due to the following argument.
  By removing the edge $uv$ from $G$ we obtain two graphs $G_1, G_2$; let us consider the graphs $G_1^+ := G_1 \cup \{uv\}$ and $G_2^+ := G_2 \cup \{uv\}$.
  Notice that $uv$ is a pendant edge in both graphs.
  According to \Cref{lem:optimal_pendant_unique} we may assume that $G_1^+$ and $G_2^+$ have optimal colorings $\chi^\star_1$ and $\chi^\star_2$, respectively, such that $|\chi^\star_i(G_i^+)| = |\chi^\star(G_i^+)|$ holds for $i=1,2$ and $uv$ has a unique color in both: denote this color by~$c_1$ and $c_2$, respectively.
  Now, identifying $c_1$ and $c_2$, taking the colors of $\chi^\star_1$ and $\chi^\star_2$ we can obtain a coloring for $G$, in which the bridge $uv$ has a unique color and which uses $|\chi^\star|$ many colors.
  
  Let $G_1'$ and $G_2'$ (respectively) denote $G_1$ and $G_2$ (respectively) after Modification~4; that is, replacing edges $u_1u$ and $u_2u$ by $u_1u_2$ and replacing edges $v_1v$ and $v_2v$ by $v_1v_2$, respectively.
  In the case with two copies of edges $u_1u_2$ or $v_1v_2$, we use the same argument as above.
  Then, the modified graph $G'$ consists of 3 components: $G_1'$, $G_2'$ and the edge $uv$.
  
  As for the colorings of $G$ and $G'$, we can use the arguments of {Case 2}.  
  After the modification, all vertices of $G'$ see the same colors as in $G$ and $|\chi'| = |\chi^\star|$ holds and hence $\opt(G) \leq \opt(G')$ is true.
  In the other direction, a coloring $\chi$ for~$G$ can be obtained from a coloring $\chi'$ for $G'$, such that $|\chi| = |\chi'|$. 
  Therefore, Modification~4 is valid. 
 \end{proof}

\subcubic*

\begin{proof}
  Let $G$ be a subcubic graph.
  Apply Modifications 1--4, until they do not change the graph anymore,
  and denote the resulting graph by $G'$.
  We prove that all components of $G'$ are either trivial or they contain only degree-$3$ vertices.
  
  Observe that Modifications 1--4 cannot increase the degree of any vertex.
  Note that $G'$ does not contain vertices of degree 2, due to Modification~2.
  Now consider vertices of degree 1.
  Since we remove pendant edges via Modification~4, the only vertices of degree 1 are the end vertices of independent edges, which are trivial components.
  Vertices of degree 0 are also trivial components.

  After applying Modifications 1--4, each component is therefore either a trivial component or a bridgeless cubic graph.
  It is well-known that each bridgeless cubic graph has a perfect matching~\cite{Petersen1891}.
  Using the equivalence of modified graphs due to \Cref{lem:mod_approx}, the claim follows as a consequence of \Cref{thm:normalized_perfect_matching}.
 \end{proof}

In order to prove our result on claw-free graphs, we need to introduce a new modification that helps us control the number of claws.

\paragraph*{Modification 5: Avoid neighboring pendant edges}

Suppose there are two adjacent vertices in $G$, denoted by $u_1$ and $u_2$, such that both of them have exactly one adjacent pendant vertex, denoted by $v_1$ and $v_2$, respectively (see \Cref{fig:modification5}).
Modification~5 contracts the edge $u_1 u_2$ into a new vertex $u_{12}$ with exactly one pendant edge $u_{12} v_{12}$ incident. 
Then $u$ \enquote*{inherits} all other neighbors of $u_1$ and $u_2$, without multiplicities.
Furthermore, we introduce an isolated edge $w_1 w_2$.

\begin{lemma}\label{lem:mod5}
Modification~5 is valid.
\end{lemma}

\begin{proof}
According to \Cref{lem:optimal_pendant_unique}, we can assume that we have an optimal coloring $\chi^\star$ which assigns unique colors to pendant edges.
Then $\chi^\star(u_1 v_1) = c_1$ and $\chi^\star(u_2 v_2) = c_2$, such that $c_1 \neq c_2$ and no other edge is colored $c_1$ or $c_2$.
But then $\chi^\star(u_1 u_2) = c_3$, where $c_3 \notin \{c_1, c_2\}$. 
Every other edge adjacent to $u_1$ or $u_2$ (besides $u_1 v_1$ and $u_2 v_2$) is colored $c_3$. 

After Modification~5, we construct a coloring $\chi'$ for the modified graph $G'$ as follows.
For every possible $v \neq v_{12}$, each existing edge $u_{12} v$ gets color $c_3$, $u_{12} v_{12}$ gets color $c_1$ and the isolated edge $w_1 w_2$ gets color $c_2$.
Observe that this is a feasible coloring $\chi'$ for $G'$ that uses the same number of colors as $\chi^\star$, hence $\opt(G) \leq \opt(G')$.

In the other direction, suppose we have a graph $G'$ with a pendant edge $u_{12} v_{12}$ at $u_{12}$, and an isolated edge $w_1 w_2$.
Using \Cref{lem:pendant_unique}, for a coloring $\chi'$ of $G'$ we can construct a coloring $\hat\chi'$ that assigns $u_{12} v_{12}$ a unique color~$c_1$ and $|\chi'| \leq |\hat\chi'|$.
Then, for each $v \neq v_{12}$, if the edge $u_{12} v$ exists it is colored $c_3 \neq c_1$.
Additionally, $\hat\chi'(w_1 w_2) = c_2$ for some unique color $c_2$.

Let us now construct a graph $G$ by splitting vertex $u_{12}$ into two adjacent vertices $u_1$ and~$u_2$, and replacing some vertices $u_{12} v$ with $u_1 v$ or $u_2 v$, or both, inheriting the color $c_2$ of $u_{12} v$.
Furthermore, replace $u_{12} v_{12}$ by $u_1 v_1$, and replace $w_1 w_2$ by $u_2 v_2$, and inherit the respective colors, $c_1$ and $c_2$.
Finally, use the color $c_3$ for the edge $u_1 u_2$.
Observe that (regardless to possible equivalences between colors $c_1$, $c_2$ and $c_3$)  this leads to a coloring $\chi$ feasible for~$G$, with the same amount of colors than $\hat\chi'$, so then $|\chi'| \leq |\hat\chi'| = |\chi|$ holds, finishing the proof.
 \end{proof}

\begin{figure}[tb]
  \centering
  \includegraphics[scale=0.38]{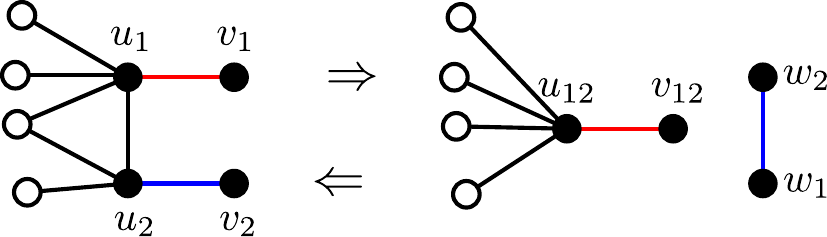}
  \caption{Modification~5, with $c_1$, $c_2$ and $c_3$ marked with red, blue and black edges, respectively. \label{fig:modification5}}
\end{figure}

Now we are ready to handle claw-free graphs.

\clawfree*

\begin{proof}
  Let $G$ be a claw-free graph.
  We will show that applying Modifications 1--3 either does not create a claw, or creates a claw in a controlled manner.
  
  \textbf{Keeping it (almost) claw-free.}
  Modification~1 removes pendant vertices, so it cannot create a claw.
  Suppose $v$ is a degree-$2$ vertex in a claw-free graph $G$, with two adjacent edges $u_1 v$ and $u_2 v$.
  Modification~2 splits $v$ into $v_1$ and $v_2$, and replaces the two incident edges by $u_1 v_1$ and $u_2 v_2$.
  Since $G$ did not contain a claw, the only way $G'$ can contain a claw is if there was a vertex $x \notin \{u_1, u_2, v\}$ and edges $u_1 u_2$ and $u_1 x$ (or analogously $u_2 x$), but no edge $u_2 x$ (analogously no $u_1 x$).
  Since $v$ has a degree of $2$, the vertices $u_1, u_2, v_1, x$ induce a claw in $G'$.

  Observe that the pendant edges $u_1 v_1$ and $u_2 v_2$ are neighbors, therefore we can apply Modification~5, that is, contract $u_1 u_2$ into $u_{12}$ and replace edges $u_1 v_1$ and $u_2 v_2$ by $u_{12} v_{12}$ and $w_1 w_2$.
  Furthermore, we replace all existing edges of form $u_1 v$ and $u_2 v$ by an edge $u_{12} v$, and denote this graph $G''$.
  Observe that by contracting the edge $u_1 u_2$, Modification~5 has removed the claw $u_1, u_2, v_1, x$.
  It is possible, however, that the newly introduced edge $u_{12} v_{12}$ participates in a claw $\{u_{12}, v_{12}, y_1, y_2\}$, for some vertices $y_1$, $y_2$.
  Note that this means there was no edge $y_1 y_2$ in $G$, and therefore $G'$, and additionally, $y_1$ and $y_2$ were adjacent to $u_1$ and~$u_2$ in $G$, respectively.
  Indeed, if both $y_1$ and $y_2$ were adjacent to $u_1$ (or $u_2$), this would lead to the claw $\{u_1, u_2, y_1, y_2\}$ in $G$, contradicting its claw-freeness.
  Therefore, Modification~2 creates at most one claw, but when it does it also creates a pendant edge and an isolated edge in the process, both of which will be accounted for in the second part of the proof.
  Note that $u_{12} v_{12}$ does not have a neighboring pendant edge, as that would mean a pendant edge at $u_1$ or $u_2$, which would contradict the claw-freeness of $G$; hence no pendant edges will be removed after creating the claw.
  
  Modification~3 removes a simple cactus $C$ but keeps its
  {needles} in $G$.
  Therefore the only way it can create a claw is if two needles of $C$, say $u_1 v$ and $u_2 v$ are  incident to the same vertex~$v$ in $G \setminus C$: the claw is induced by $v$, $u_1$, $u_2$ and $u_3$ where $u_3$ is a vertex adjacent to~$v$.%
  \footnote{Note that before the modification the edge $u_1 u_2$ prevented the claw.}
  But then, after Modification~3 has been applied, $v$ has two adjacent pendant vertices~$u_1$ and~$u_2$, and one of them, say $u_2$, can be removed via Modification~1.
  This removes the claw as well; if there were another vertex $x$ such that $u_1$, $u_3$, $v$, $x$ is a claw in the modified graph, these vertices would induce a claw in the original graph as well, a contradiction.

  \textbf{Bounding the approximation factor.}   
  As a result, we arrive to a normalized graph~$G'$ with possibly more than one component.
  Consider first the case when $G'$ is claw free; we argue about each component $C$ independently.
  Since $C$ is normalized, the optimal coloring of $C$ has at most $3|C|/4$ colors.
  If $|C|$ is even, it contains a perfect matching, which is of size $|C|/2$ due to Sumner~\cite{Sumner1974} and Las Vergnas~\cite{LasVergnas1975} and the claw-freeness of $G$.
  If $|C|$ is odd, we show that it contains a matching of size $(|C|-1)/2$.
  Remove a non-cutvertex $v$ from~$C$ and denote the resulting connected graph $C'$.

  Note that the removal of $v$ does not create a claw in~$C'$, hence it has a matching of size $|C'|/2 = (|C|-1)/2$~\cite{Sumner1974,LasVergnas1975}.
  As a consequence, there is a matching of size $(|C|-1)/2$ in $C$.
  In both cases, there is a matching of size $(|C|-1)/2$, so \Cref{alg:basic} outputs a solution of size%
  \footnote{Note that the only case when we do not have the one additional color as per Step~3 of \Cref{alg:basic} is when $E(C) \setminus E(M)$ is empty, but then the graph $C$ itself is $|C|/2$ independent edges, which has an approximation ratio of $1$.}
  at least $(|C|-1)/2+1 = |C|/2+1/2$, leading to an approximation factor at most $\frac{3|C|/4}{|C|/2+1/2} < 3/2$ for every component $C$ in $G'$, and hence the claim follows. 
  
  Now let us consider the case when claws are being created during the normalization process.
  As shown in the first half of the proof, performing Modification~2 then Modification~5 results in creating a claw, but it also creates a new pendant edge and a new isolated edge in the process.
  Let us isolate the creation of one single claw, we will see that this argument can be repeated for all claws created.
  During Modification~2 and Modification~5 we remove a vertex from $G$ and add an isolated edge to it, which means the number of vertices is increased by $1$.
  On the other hand, we introduced $3$ new pendant vertices: $v_{12}$, $w_1$ and~$w_2$.
  According to \Cref{lem:maxcolors}, the number of colors in a feasible coloring is therefore at most $3(n+1)/4 - 3/4 = 3n/4$.
  Observe that the size of the maximum matching did not change: it decreased by $1$ due to the contraction, but increased by $1$ due to the isolated edge; therefore \Cref{alg:basic} still finds a coloring with at least $n/2$ colors (depending on the parity of $n$).
  Creating a claw this way thus always results in a graph, for which \Cref{alg:basic} is a $3/2$-approximation.
  Since the pendant edge and the isolated edge do not get removed by any modification in a later step of the algorithm, we can simply repeat the same argument for every claw created. Thus, the claim follows.  
 \end{proof}

\section{$1.625$-approximation for Graphs with Perfect Matching}
\label{sec:pm}

\begin{definition}
  A modification is \emph{perfect matching preserving}, if for each graph $G$ that has a perfect matching $M$, the modification generates a graph $G' \equiv G$ such that $G'$ also has a perfect matching.
\end{definition}

\begin{lemma}\label{lem:pm_preserving}
  Modification~1 and Modification~3 are perfect matching preserving.
\end{lemma} 
\begin{proof}
We first show that Modification~1 is perfect matching preserving.
Observe that if there are adjacent pendant edges $uv$ and $uw$ in $G$, only one of the vertices $v$ and $w$ can be matched, therefore there is no perfect matching in $G$.
Otherwise Modification~1 does not have any effect.
Therefore, it is trivially perfect matching preserving.

    \vspace{0.5em}

  In order to ensure that Modification~3 is perfect matching preserving, we have to apply some variation to the modification.
  Let us fix a perfect matching~$M$ in $G$ and start with a simple cactus $C$ consisting of just one triangle $T = \{u, v, w\}$.
  Let us denote the needles of $T$ by $e_u$, $e_v$ and $e_w$.
  If $e_u$, $e_v$ and $e_w$ are in $M$, then just perform Modification~3, i.e. replace the triangle $u v w$ by an edge $x y$.
  In the resulting graph $G'$, vertices $u$, $v$ and $w$ will be still matched by $e_u$, $e_v$ and $e_w$, respectively, and $M$ can be extended with the edge $x y$ to cover the newly added vertices~$x$ and $y$.
  
  Suppose now that $T$ contains an edge from $M$, w.l.o.g. assume that it is the edge $u v$.
  In this case, replace the triangle $\{u, v, w\}$ not with an edge $x y$ but with the edge $u v$, illustrated on \Cref{fig:triangle}c) in \Cref{sec:algorithm}.
  Note that the edge can be still colored with a new color, as $u$ and~$v$ will only have a single adjacent edge in $G'$,~$e_u$ and $e_v$, respectively. 
  Moreover, there were no new vertices added, and there were no edges that participate in the matching~$M$ removed, so the modified graph $G'$ also has a perfect matching.
  Notice that this variant of Modification~3 is still valid: in the latter case the edge $u v$ simply takes the role of $x y$, as it can still hold the color of $T$ in the modified graph~$G'$.   
  
  We follow the same approach for cacti $C$ with more than one triangle.
  In case a perfect matching $M$ exists in $G$, each triangle $T_i = \{u_i, v_i, w_i\}$ of $C$ will either (i) contain one edge from~$M$ (w.l.o.g. $u_i v_i$), and have one vertex matched to outside $T_i$; or (ii) have all three of its vertices matched to outside $T_i$.
  In case (i), we keep the edge $u_i v_i$, whereas either $w_i$ is incident to a needle or the neighboring triangle keeps its edge containing $w_i$, therefore all vertices of $T_i$ remain matched.
  In case (ii), we add an edge $x_i y_i$ and keep the vertices $u_i$, $v_i$ and $w_i$ matched to outside of $T_i$, either in another triangle (where it is matched due to case (i)) or outside the cactus $C$ (where it is matched due to it being matched in $G$).
  In this case, the new edge $x_i y_i$ will be added to the perfect matching, and no vertex will become unmatched.  
  
  Using the same arguments as in the proof of \Cref{lem:mod3}, we can assume that each triangle is monochromatic.
  Then, similarly to the single triangle case, there is an edge in $G'$ for each triangle in $G$ to allow for calculating a feasible coloring $\chi'$ from $\chi$ of equal size, and vice versa.
 \end{proof}

  We remark that coming up with a perfect matching preserving modification that removes degree-$2$ vertices would yield a $1.5$-approximation for graphs that contain a perfect matching.
  Indeed, in that case the arguments in the proof of \Cref{thm:normalized_perfect_matching} would give us the result, as we could apply Modifications~1--3 to a graph containing a perfect matching until we obtain a normalized graph with a perfect matching.
  As we do not have a perfect matching preserving Modification~2, we use another  approach, yielding a slightly worse approximation factor.

  Suppose $G$ has a perfect matching, then Modification~2 introduces a new vertex, making the number of vertices odd, hence the resulting graph does not have a perfect matching anymore.
  This affects the approximation ratio of \Cref{alg:basic}: the number of vertices increases by one, the number of leaves by two, while the size of the maximum matching stays the same.
  One can show that although the approximation ratio can get worse, it does not get worse than $13/8 = 1.625$.
  Thus, we obtain~\Cref{thm:perfect_matching}.

\perfectmatching*

\begin{proof}
  Let $G$ be a graph that contains a perfect matching; let us fix such a matching $M$ for the rest of the proof.
  According to \Cref{lem:pm_preserving}, Modification~1 and~3 do not violate the property of having a perfect matching.
  Therefore if we only perform these two modifications and arrive at a normalized graph, \Cref{alg:basic} gives a $1.5$-approximation according to \Cref{thm:normalized_perfect_matching}.
  We will now analyze how Modification~2 can affect the performance of \Cref{alg:basic} using arguments similar to those in the case of claws in the proof of \Cref{thm:clawfree}.
  
  Recall that Modification~2 splits a degree-$2$ vertex $v$ into two new vertices, $v_1$ and $v_2$, replacing the edges $u_1 v$ and $u_2 v$ by $u_1 v_1$ and $u_2 v_2$, respectively.
  This operation (let us call them individually as \textit{events}) increases the total number of vertices in $G$ by one and creates two leaves; however, it might create neighboring pendant edges, which also get simplified by Modification~1.
  
  \textbf{Upper bound on the number of colors.}
  Let us analyze the changes in the number of pendant edges and vertices, in all possible cases.
  \\ \textbf{Case 1:} no pendant edge gets removed, therefore their number grows by $2$.
  Then the number of vertices is increased by $1$ and the upper bound is $3/4 (n+1) - 1/2 = 3/4 n + 1/4$, by \Cref{lem:maxcolors}.
  \\ \textbf{Case 2:} one of the pendant edges gets removed, therefore their number grows by $1$.
  Then the number of vertices stays $n$, hence the upper bound is $3/4 n - 1/4$.
  Notice that in this case we had a pendant edge $u_1 w$ (or $u_2 w$) in $G$, which was part of $M$ along with $u_2 v$ (or $u_1 v$), and the newly created $u_1 v_1$ (or $u_2 v_2$) got removed by Modification~1.
  \\ \textbf{Case 3:} both pendant edges get removed, so their number does not change.
  Then the number of vertices decreases by $1$, so the upper bound is $3/4 n - 3/4$.
  Note that this would mean both $u_1 v_1$ and $u_2 v_2$ have neighboring pendant edges, say $u_1 w_1$ and $u_2 w_2$. 
  But this could only be a result of Case 1: $G$ originally contained a perfect matching, and either $u_1v$ or $u_2v$ was a matching edge. 
If $u_iv$ was the matching edge, the unmatched pendant edge $u_iw_1$ must have been created by Case $1$.
  Let us couple these two individual events, and later argue that \textit{in total} the number of vertices stays the same, but the number of pendant edges is increased by $2$. Observe that since $u_i$ only has one incident matching edge, $u_i v$, this particular Case~$1$ event will not be coupled to any other Case~$3$ event. 
  
   Following the arguments above, we can see that for each time Case 3 happens there is an occasion where Case 1 happened. 
  As a result, the upper bound after these two coupled events becomes $3/4 n - 1/2$, which means it decreases by $1/4$ \textit{per event}.
  Let us do the accounting the following way.
  We have Case 3, coupled with Case 1, that decrease the bound by $1/4$ for each time we perform Modification~2.  
  The same holds for Case 2; let us denote the total number of such instances of Modification~2 by $d_2^-$.
  On the other hand, let us denote the number of Case 1 events \textit{without} a following Case 3 event by $d_2^+$; we have shown that these events increase the bound by $1/4$.

  In total, after performing all possible events of Modification~2, the upper bound on the number of colors hence becomes $3/4 n + 1/4 d_2^+ - 1/4 d_2^-$.
  
  \textbf{Two lower bounds on the number of colors returned by \Cref{alg:basic}.} 
  Note that Modification~2 kept all matching edges, in all three cases, and every pendant edge removal that follows will also not remove a matching edge: if there are multiple pendant edges adjacent to a common vertex, at most one will be part of a matching, and we simply choose to keep that pendant edge.
  This means that the size of the maximum matching did not decrease, hence \Cref{alg:basic} will output a coloring of size at least $n/2+1$.
  The number of vertices and thus the bound on the size of the optimal solution, however, may have increased.
  
  Fortunately, there is another way we can bound the size of the maximum matching from below.
  Notice that the pendant edges introduced by Modification~2 are independent, i.e. they form a matching.
  Using the same case distinction as before, there are $d_2^+$ events that increase the number of pendant edges by $2$, and $d_2^-$ events that increase it by $1$.
  This means that in the modified graph $G'$ we can always find a matching of size $2 d_2^+ + d_2^-$, simply by choosing the newly created pendant edges.
  
  \textbf{Combining the bounds.}
  Assume first that we introduce with Modification~2 at most $n/2$ pendent edges, that is, $2 d_2^+ + d_2^- \leq n/2$~$(*)$ holds.
  Then the approximation ratio of \Cref{alg:basic} is at most
  \begin{equation}
  \frac{ \frac34 n + \frac14 d_2^+ - \frac14 d_2^- }{ \frac{n}{2} } \stackrel{(*)}{\leq}
  \frac{ \frac34 n + \frac{1}{16} n - \frac38 d_2^-}{ \frac{n}{2} } \leq
  \frac{ \frac{13}{16} n }{ \frac{n}{2} } = \frac{13}{8} \enspace .
  \end{equation}
  
  Assume now that we introduce with Modification~2 at least $n/2$ pendent edges, that is, $2 d_2^+ + d_2^- \geq n/2$~$(**)$ holds.
  Then the approximation ratio of \Cref{alg:basic} is at most
  \begin{equation}
  \frac{ \frac34 n + \frac14 d_2^+ - \frac14 d_2^- }{ 2 d_2^+ + d_2^- } \leq
  \frac{ \frac34 n }{2 d_2^+ + d_2^-} +  \frac{ \frac14 d_2^+ - \frac14 d_2^- }{ 2 d_2^+ + d_2^- } \stackrel{(**)}{\leq}
  \frac{ \frac34 n}{\frac{n}{2}} + \frac{ \frac14 d_2^+ }{ 2 d_2^+ } =
  \frac32 + \frac18 = \frac{13}{8} \enspace ,
  \end{equation}
  finishing the proof.      
 \end{proof}

\end{document}